\documentclass{article}
\usepackage{kr}

\usepackage[table]{xcolor}
\usepackage{thmtools}
\usepackage{thm-restate}

\usepackage{times}
\usepackage{soul}
\usepackage{cancel}
\usepackage[small]{caption}
\usepackage{booktabs}
\usepackage{algorithm}
\usepackage{algorithmic}
\usepackage{xspace}

\usepackage{eucal}
\usepackage{multicol}

\usepackage[textwidth=1.5cm]{todonotes}
\usepackage[utf8]{inputenc} %
\usepackage[T1]{fontenc}    %
\usepackage{hyperref}       %
\usepackage{url}            %
\usepackage{booktabs}       %
\usepackage{amsfonts}       %
\usepackage{nicefrac}       %
\usepackage{microtype}      %
\usepackage{lipsum}
\usepackage{fancyhdr}       %
\usepackage{graphicx}       %
\graphicspath{{media/}}     %
\usepackage{wrapfig}        %

\usepackage{complexity}
\mathchardef\mhyphen="2D %
\usepackage{arydshln}
\usepackage{xcolor}
\usepackage[inline]{enumitem}
\usepackage{bigstrut}
\usepackage{comment}
\usepackage{tikz}
\usepackage{adjustbox}
\usetikzlibrary{positioning} 
\usepackage{stmaryrd}
\usepackage{amssymb,amsmath,amsthm,mdsymbol}
\usepackage{tabularx}
\newtheorem{fact}{Fact}

\usepackage[switch,mathlines]{lineno}
\usepackage{etoolbox}

\usepackage{cleveref}

\newcommand*\linenomathpatch[1]{%
  \cspreto{#1}{\linenomath}%
  \cspreto{#1*}{\linenomath}%
  \csappto{end#1}{\endlinenomath}%
  \csappto{end#1*}{\endlinenomath}%
}
\newcommand*\linenomathpatchAMS[1]{%
  \cspreto{#1}{\linenomathAMS}%
  \cspreto{#1*}{\linenomathAMS}%
  \csappto{end#1}{\endlinenomath}%
  \csappto{end#1*}{\endlinenomath}%
}

\expandafter\ifx\linenomath\linenomathWithnumbers
  \let\linenomathAMS\linenomathWithnumbers
  \patchcmd\linenomathAMS{\advance\postdisplaypenalty\linenopenalty}{}{}{}
\else
  \let\linenomathAMS\linenomathNonumbers
\fi

\linenomathpatch{equation}
\linenomathpatchAMS{gather}
\linenomathpatchAMS{multline}
\linenomathpatchAMS{align}
\linenomathpatchAMS{alignat}
\linenomathpatchAMS{flalign}

\makeatletter
\patchcmd{\mmeasure@}{\measuring@true}{
  \measuring@true
  \ifnum-\linenopenaltypar>\interdisplaylinepenalty
    \advance\interdisplaylinepenalty-\linenopenalty
  \fi
  }{}{}
\makeatother

\newtheorem{theorem}{Theorem}
\newtheorem{corollary}[theorem]{Corollary}
\newtheorem{lemma}[theorem]{Lemma}

\newtheorem{proposition}[theorem]{Proposition}
\newtheorem{claim}[theorem]{Claim}

\theoremstyle{definition} %
\newtheorem{definition}[theorem]{Definition}
\newtheorem{example}[theorem]{Example}

\newcommand{\chasenext}[3]{#1 \xrightarrow{\ #3\ } #2}
\newcommand{\Rp}{\ensuremath{\mathbb{R}_{\geq 0}}\xspace}%
\newcommand{\Att}{\mathrm{Att}}

\newcommand{\momodels}[1]{\models_{#1}}
\newcommand{\bmomodels}[1]{\models^{\rm b}_{#1}}
\newcommand{\relvdash}[1]{\vdash_{#1}}
\newcommand{\wsvdash}{\relvdash{\rm ws}}
\newcommand{\svdash}{\relvdash{\rm s}}

\newcommand{\mD}{\mathfrak{D}}

\newcommand{\tup}{\mathrm{Tup}}

\newcommand{\Dom}{\mathrm{Dom}}

\newcommand{\sch}{\mathbf{S}}

\newcommand{\cl}[1]{{#1}^{*,\mathrm{ws}}}
\newcommand{\ordcl}[1]{{#1}^{*}}
\newcommand{\scl}[1]{{#1}^{*,\mathrm{s}}}

\newcommand{\KX}{R} 
\newcommand{\KY}{S} 
\newcommand{\KZ}{T} 
\newcommand{\Const}{\mathrm{Const}}
\newcommand{\Rel}{\mathrm{Rel}}
\newcommand{\supp}{\mathrm{Supp}}

\newcommand{\inv}[1]{{#1}^{-1}}

\newcommand{\marg}[2]{{#1}[{#2}]}

\newcommand{\dfn}{\coloneqq}

\newcommand{\mplus}{\oplus}
\newcommand{\bigmplus}{\bigoplus}

\newcommand{\ar}{\mathrm{ar}}

\newcommand{\inc}[4]{{#1}[#2]\leq {#3}[#4]}
\newcommand{\Mo}{\mathbb{K}}
\newcommand{\DMo}{K}
\newcommand{\B}{\mathbb{B}}
\newcommand{\N}{\mathbb{N}}

\newcommand{\adom}{\mathrm{Adom}}

\newcommand{\chase}[2]{\mathrm{Chase}({#1},{#2})}
\newcommand{\chaseplus}[2]{\mathrm{Chase}^{\oplus}({#1},{#2})}

\newcommand{\rel}[1]{\ensuremath{\mathsf{#1}}}
\newcommand{\attr}[1]{\ensuremath{\mathit{#1}}}

\newcommand{\defeq}{\mathrel{\mathop:}=}
\title{Dichotomy for Axiomatising Inclusion Dependencies on K-Databases}

\author{%
Miika Hannula$^1$\and
Teymur Ismikhanov$^1$\and
Jonni Virtema$^{2}$\and
\affiliations
$^1$Institute of Computer Science, University of Tartu, Estonia\\
$^2$School of Computing Science, University of Glasgow, UK\\
\emails
\{miika.hannula, teymur.ismikhanov\}@ut.ee,\\
jonni.virtema@glasgow.ac.uk.
}

\begin{document}

\maketitle

\begin{abstract}
    A relation consisting of tuples annotated by an element of a monoid $\Mo$ is called a $\Mo$-relation. A $\Mo$-database is a collection of $\Mo$-relations. 
    In this paper, we study entailment of inclusion dependencies over $\Mo$-databases, where $\Mo$ is a positive commutative monoid.
    We establish a dichotomy regarding the axiomatisation of the entailment of inclusion dependencies over $\Mo$-databases, based on whether the monoid $\Mo$ is weakly absorptive or weakly cancellative. We establish that, if the monoid is weakly cancellative then the standard axioms of inclusion dependencies are sound and complete for the implication problem. If the monoid is not weakly cancellative, it is weakly absorptive and the standard axioms of inclusion dependencies together with the weak symmetry axiom are sound and complete for the implication problem. In addition, we establish that the so-called balance axiom is further required, if one stipulates that the joint weights of each $\Mo$-relation of a $\Mo$-database need to be the same; this generalises the notion of a $\Mo$-relation being a distribution. In conjunction with the balance axiom, weak symmetry axiom boils down to symmetry.\looseness=-1
\end{abstract}

\section{Introduction}

In the past two decades, there has been an increasing interest in the study of logical formalisms and foundations of database theory in the settings of semiring or monoid annotations. The initial motivation for this work can be traced to the influential article \cite{DBLP:conf/pods/GreenKT07}, where $\Mo$-relations were introduced as a way to model provenance in databases. For a semiring $\Mo$, a $\Mo$-relation can be seen as a relation whose tuples are annotated by elements of $\Mo$---for example, $\N$-relations correspond to bags while $\B$-relations are just normal relations. 
While the original framework applied only to the provenance of queries expressible in negation-free first-order logic, subsequent investigations extended the study of provenance to richer languages, including full first-order logic \cite{Grädel2025}, Datalog \cite{DBLP:conf/icdt/DeutchMRT14}, and least fixed-point logic \cite{DBLP:conf/csl/DannertGNT21}.

The general setting of $\Mo$-relations has lead to a multitude of logical formalisms to be invented to utilise this abstract framework past its original motivation of provenance analysis.
Most of the works relate to studying properties of queries such as conjunctive query containment \cite{DBLP:journals/mst/Green11,DBLP:journals/tods/KostylevRS14}, 
the algebraic semantics and evaluation of Datalog queries \cite{DBLP:journals/jacm/KhamisNPSW24,DBLP:journals/pacmmod/ZhaoDKRT24}, and ontology-based data access \cite{CalvaneseLOP019,BourgauxOPP20}.

A complementary approach, which is also taken in the current paper, is to focus on database integrity constraints (ICs) in the setting of $\Mo$-databases.
The combination of $\Mo$-databases and ICs has very recently raised interest in multiple fronts. \citeauthor{KR2025-32} \shortcite{KR2025-32} utilise $\Mo$-databases to formalise a general logic based framework for database repairs and consistent query answering, \citeauthor{DBLP:journals/pacmmod/KolaitisPVW25} \shortcite{DBLP:journals/pacmmod/KolaitisPVW25} focus on the complexity of consistent query answering of self-join-free conjunctive queries over $\Mo$-databases, and \citeauthor{Hannula24a} \shortcite{Hannula24a} considers conditional independence and normalization over $\Mo$-relations.
Often monoids or semirings offer a convenient mathematical framework to study how far results that have been proven in the set and bag semantics can be generalised---what are the properties of the Boolean semiring or the semiring of natural numbers that is needed for the results to hold; e.g, \cite{DBLP:journals/jacm/AtseriasK25,BHKPV26} study the interplay between local consistency and global inconsistency, and base their characterisations on $\Mo$-relations over positive commutative monoids. The latter also considers axiomatisability of functional dependencies in the setting of locally consistent families of $\Mo$-relations. This is particularly relevant to the current paper, since the local consistency property can be expressed by using inclusion dependencies.  

The current paper focuses on axiomatising the entailment of inclusion dependencies (INDs) on $\Mo$-databases, where $\Mo$ is a naturally order positive commutative monoid.
An \emph{inclusion dependency} is an expression of the form $\inc{R}{A_1,\dots ,A_n}{S}{B_1,\dots ,B_n}$, where $R$ and $S$ are relation names, and $A_1,\dots ,A_n$ and $B_1,\dots ,B_n$ are pairwise distinct attribute sequences. Intuitively, a relational database satisfies the
above dependency,
if every sequence of data values that occur in $R$ for the attributes $A_1,\dots ,A_n$ also occur in $S$ as data values for the attributes $B_1,\dots ,B_n$. Over bag databases, we further demand that each sequence of data values occur in $S$ with greater or equal multiplicity than they occur in $R$. This is generalised to ordered monoids in the obvious manner. 

\begin{example}[Budgets and expenses]\label{ex:budget}
Consider an annotated database depicted in \Cref{tab:expense}, where the annotations represent monetary amounts. The $\Mo$-relation $\rel{Expense}(\attr{proj,type,year})$ records spending
items, $\rel{Budget}(\attr{proj,year})$ annual budgets, and
$\rel{Grant}(\attr{proj})$ the total grant available per project.

The following INDs express that annual spending stays within the annual budget,
and that the total budgeted amount matches exactly the available grant:
\begin{enumerate}[
  label=(C\arabic*),
  ref=C\arabic*,
  leftmargin=*,
  labelsep=0.6em
]
\item \label{con:exp-bud}
$\rel{Expense}[\attr{proj,year}] \leq \rel{Budget}[\attr{proj,year}]$

\item \label{con:bud-grant}
$\rel{Budget}[\attr{proj}] \leq \rel{Grant}[\attr{proj}]$

\item \label{con:grant-bud}
$\rel{Grant}[\attr{proj}] \leq \rel{Budget}[\attr{proj}]$
\end{enumerate}
\end{example}

\begin{table}[t]
    \footnotesize
  \centering
  \setlength{\tabcolsep}{6pt}
\begin{tabular}{l}

\begin{tabular}{@{}l l l: l@{}}
\multicolumn{4}{l}{$\rel{Expense}$}\\
\toprule
\attr{proj} & \attr{type} & \attr{year} & \# \\
\midrule
P1 & equipment & 2024 & 1200 \\
P1 & hotel     & 2024 & 800 \\
P1 & travel    & 2025 & 600 \\
P2 & travel & 2024 & 1500 \\
\bottomrule
\end{tabular}

\\[10ex]

\begin{tabular}[t]{@{}l l: l@{}}
\multicolumn{3}{l}{$\rel{Budget}$}\\
\toprule
\attr{proj} & \attr{year} & \# \\
\midrule
P1 & 2024 & 2500 \\
P1 & 2025 & 1500 \\
P2 & 2024 & 2000 \\
\bottomrule
\end{tabular}

\quad

\begin{tabular}[t]{@{}l: l@{}}
\multicolumn{2}{l}{$\rel{Grant}$}\\
\toprule
\attr{proj} & \# \\
\midrule
P1 & 4000 \\
P2 & 2000 \\
\bottomrule
\end{tabular}

  \end{tabular}
  \caption{Project expenses and budgets.\label{tab:expense}}
\end{table}

For a finite set of INDs $\Sigma\cup\{\sigma\}$ and a naturally ordered positive commutative monoid $\Mo$, we say that $\Sigma$ entails $\sigma$ on $\Mo$-databases (and write $\Sigma\models_{\Mo} \sigma$), if every $\Mo$ database that satisfies every IND in $\Sigma$ also satisfies $\sigma$.

\vspace{2mm}
\noindent \textbf{Our contributions.}
Our starting points are the sound and complete axiomatisations of the entailment of INDs over standard relational databases
\cite{CasanovaFP84} and over uni-relational $\Rp$-databases \cite{HannulaV22}. The former establishes that three axioms schemas---reflexivity, transitivity, and projection and permutation---constitute a sound and complete axiomatisation for the entailment of INDs over relational databases, while the latter establishes that by adding a fourth axiom---symmetry---one obtains a sound and complete axiomatisation for the entailment of INDs over uni-relational $\Rp$-databases.

In this paper, we show that the above two axiomatisations yield a dichotomy with respect to non-trivial naturally ordered positive commutative monoids $\Mo$. We show that if $\Mo$ is what we call weakly absorptive (see Def.~\ref{def:monoidprops} for a definition) then the three standard axioms of inclusion dependency constitute a sound and complete axiomatisation over $\Mo$-databases.  We then show that if $\Mo$ is not weakly absorptive---equivalently, if it is what we call weakly cancellative---the three standard axioms of inclusion dependency together with the so-called \emph{weak symmetry} axiom constitute a sound and complete axiomatisation over $\Mo$-databases. Finally, we establish that, if we restrict to so-called \emph{balanced} databases, an additional balance axiom is required. A database is balanced if the marginalisation of each relation to the empty set of attributes coincide. This notion generalises the notion of probability distributions, and includes the case of uni-relational databases. See Table~\ref{table:results} for an overview of our results.

In proving our results, we also develop an additive variant of the chase procedure, called the $\mplus$-chase, which operates directly on monoid-annotated databases. In contrast to the classical chase, each chase step increases the weight of a single tuple just enough to repair a violated  inequality. The procedure is not guaranteed to terminate in general, however termination is ensured whenever the set of inclusion dependencies is closed under weak symmetry, yielding a finite canonical model that characterises entailment. This provides a novel algebraic variant of the chase with potential independent interest.

\vspace{2mm}
\noindent \textbf{Related work.}
The contemporary manuscript by \citeauthor{Hirvonen26} \shortcite{Hirvonen26} studies axiomatisations of various dependencies in the setting of $\Mo$-team semantics, i.e., in the setting of uni-relational $\Mo$-databases. Most relevant to the current paper are their axiomatisations of unary inclusion dependencies in conjunction with functional dependencies. Furthermore, \cite[Theorem 6]{Hirvonen26} establishes a weaker uni-relational analogue of our \Cref{thm:sa}.
Lastly, we note that like algebraic Datalog evaluation, our $\mplus$-chase is also a monotone fixpoint process over an annotated database, but rather than computing query answers, it is used for enforcing integrity constraints to construct canonical databases for entailment.

\section{Preliminaries}
\begin{table}
\centering
\scalebox{0.85}{
    \begin{tabular}{@{}l| c c c | l @{}}
\toprule
        \attr{Monoid/Axiom} & \attr{IND} & \attr{WS} & \attr{BL} & \attr{Reference}  \\
\midrule
\attr{\B} & $\checkmark$    &  & & Casanova et al.~\shortcite{CasanovaFP84} \\
\attr{\Rp \,\&\, balanced} & $\checkmark$ & $\checkmark$ & $\checkmark$ & Hannula \& Virtema~\shortcite{HannulaV22}${}^\dagger$\\
\midrule
       \attr{WC} & $\checkmark$ & $\checkmark$ & & \Cref{thm:wc}  \\
       \attr{WC \,\&\, balanced} & $\checkmark$ & $\checkmark$ & $\checkmark$ & Corollary \ref{cor:wc}\\
       \attr{WA} & $\checkmark$    &  & & \Cref{thm:wa} \\
     \attr{WA \,\&\, balanced}    & $\checkmark$ &  & $\checkmark$ & Corollary \ref{cor:wab} \\
\bottomrule
    \end{tabular}
}
    \caption{Overview of our axiomatisation results on entailment of inclusion dependencies. 
    IND refers to the standard rules for inclusion dependencies (see \Cref{tab:inc-rules}), and WS and BL refer to the axioms of weak symmetry and balance (see \Cref{tab:inc-axioms-extra}).
    WC and WA refer to weak cancellativity and weak absorptivity. Balanced means restriction to the class of balanced databases. ${}^\dagger$ The result was proven for the uni-relational case and with the symmetry axiom (in place of weak symmetry and balance).}
    \label{table:results}
\end{table}

For a natural number $n$, we write \( [n] := \{1, \dots, n\} \).
We fix disjoint countably infinite sets $\Att$, $\Const$, $\Rel$ of \emph{attributes}, \emph{constants}, and \emph{relation names}.  
Each relation name $R$ is associated with a finite set of attributes $X\subseteq \Att$, denoted by $\Att(R)$.
A \emph{tuple} over  $X$ %
is a function $t\colon X\to \Const$.
We call $X$ the \emph{domain of} $t$ and denote it by $\Dom(t)$.
We write $\tup(X)$ for the set of all tuples $t \colon X \to \Const$. 
A \emph{relation} over $X$ is a finite set of tuples over $X$.

A \emph{(database) schema} is a finite set of relation names $\sch=\{R_1, \dots ,R_n\}$. A \emph{database} $D$ over $\sch$ is a set of relations instantiating $\sch$, i.e., a set $\{R_1^D, \dots ,R^D_n\}$, where $R_i^D$ is a relation over $\Att(R_i)$, for each $i\in [n]$.
We write $\adom(R^D_i)$ for the \emph{active domain} of the relation $R^D_i$, i.e., $\adom(R^D_i) \dfn \{t(x)\mid t\in R^D_i, x \in \Att(R_i)\}$.
The active domain of a database $D$, $\adom(D)$, is 
the set $\bigcup\{\adom(R_i^D) \mid R_i \in \sch\}$.

A \emph{monoid} is an algebraic structure $\Mo=(\DMo,\mplus,0)$, where $\mplus$ is associative and $0$ is the identity element of $\mplus$.
$\Mo$ is \emph{positive} if $a\mplus b=0$ entails $a=0=b$, for all $a,b\in K$, \emph{(left) cancellative} if $a\mplus b=a\mplus c$ entails $b=c$, and \emph{non-trivial} if $|K| > 1$. We associate each monoid $\Mo$ with its \emph{natural order} $\leq^\Mo$, defined by $a \leq^\Mo b \iff \exists c: a\mplus c = b$.
When the monoid is clear from the context or irrelevant, we drop the index and write $\leq$ instead of $\leq^\Mo$.
The natural order of a monoid is reflexive and transitive, meaning that it is a \emph{preorder}.
It need not in general satisfy {antisymmetry} (i.e., $a\leq b$ and $b\leq a$ imply $a=b$) nor totality (i.e., $a\leq b$ or $b \leq a$). Recall that a preorder equipped with antisymmetry is a partial order, and if it further satisfies totality, it is a \emph{total order}.

In this paper, $\Mo$ will always denote a
non-trivial positive commutative
 monoid, if not otherwise specified. 
In addition to the general case, we consider the Boolean monoid $\mathbb{B}=(\{0,1\}, \lor,0)$, and the monoids of non-negative reals $\Rp=([0,\infty), +,0)$ and natural numbers $\N=(\N, +,0)$ with their usual addition.

A \emph{$\Mo$-relation} over an attribute set $X$
 is a function 
 $\KX \colon \tup(X) \to K$ of finite support, where the \emph{support} of $\KX$ is the set $\supp(\KX) \coloneqq \{ s \in \tup(X) \mid \KX(s) \neq 0\}$.
For a set of attributes $Y\subseteq X$, the $\Mo$-relation $\marg{\KX }{ Y}$ is the \emph{marginalisation} of the $\Mo$-relation $\KX$ to the attributes in $Y$.
This means that $\marg{\KX }{Y}$ is a $\Mo$-relation $\tup(Y) \to K$ s.t.
\[
    \marg{\KX }{ Y}(t) := \bigmplus_{\substack{t' \in \tup(X) \\ t = \marg{t'}{  Y}}} \KX(t'),
\]
where $\marg{t }{ Y}$ is the projection of the tuple $t$ to $Y$ and $\bigmplus$ is the monoid's aggregate sum.
We sometimes write $\marg{\KX }{A_1\dots,A_n}(a_1,\dots,a_n)$ to denote $\marg{\KX }{A_1\dots,A_n}(t)$, where $t$ is the tuple mapping $A_i$ to $a_i$, for each $i\leq n$.
A \emph{$\Mo$-database} $\mD$ over $\sch$ is a finite set of $\Mo$-relations, one $\Mo$-relation $\KX^\mD$ for each $R\in\sch$.
The \emph{support} of $\mD$ is the database $\supp(\mD) \coloneq \{\supp(\KX^\mD) \mid R \in \sch\}$ and the \emph{active domain} of $\mD$ is inherited from of the support: $\adom(\mD) \coloneqq \adom(\supp(\mD))$.
We call $\mD$ \emph{balanced} if $\KX^\mD[\emptyset] = \KY^\mD[\emptyset]$, for every $R,S\in\sch$.
When the database $\mD$ is clear from the context or irrelevant, we often write $\KX$ instead of $\KX^\mD$.

\begin{example}
    Continuing Example \ref{ex:budget}, we observe that a $\Mo$-database satisfying the constraints \ref{con:exp-bud}-\ref{con:grant-bud}
need not be balanced, since there is no constraint enforcing that all of the annual budget have to be used. On the other hand, its restriction to the $\Mo$-relations $\rel{Budget}$ and $\rel{Grant}$ is balanced.
\end{example}

For two $\Mo$-relations $R$ and $S$ over a shared set of attributes $X$, we write $\KX \mplus \KY$ for the $\Mo$-relation $\KZ\colon \tup(X)\to K$ such that $\KZ(t) \dfn \KX(t) \mplus \KY(t)$ for all $t \in \tup(X)$.
For two $\Mo$-databases $\mD$ and $\mD'$ over a shared schema $\sch$, we define $\mD\mplus \mD'\coloneqq \{\KX^\mD\mplus \KX^{\mD'} \mid R\in \sch\}$.

The following proposition is straightforward.
\begin{proposition}\label{prop:supp}
    Let $\Mo$ be a positive monoid, and $\mD$ and $\mD'$ be $\Mo$-databases  over a shared schema. Then $\supp(\mD\mplus \mD')=\supp(\mD)\cup \supp(\mD')$.
\end{proposition}

\begin{definition}\label{def:monoidprops}
A monoid $\Mo=(\DMo, \mplus, 0)$ is \emph{weakly cancellative (WC)}, if it satisfies $a \mplus b=b \Rightarrow a=0$, for all $a,b\in \DMo$.
It is \emph{weakly absorptive (WA)}, if $a \mplus b=b$ for some non-zero $a,b\in \DMo$, and \emph{self absorptive (SA)} if $b \mplus b=b$ for some non-zero $b\in \DMo$.
It is \emph{$k$-absorptive ($kA$)}, for $k\geq 1$, if there exists $a_0,\dots,a_k \in \DMo$ such that $a_0\neq 0$ and $a_{i}\mplus a_{i + 1}=a_{i + 1}$, for all $i<k$. 
It is \emph{countably absorptive (CA)} if there are $a_0\neq 0$ and $(a_i)_{i\in \mathbb{N}_{>0}}$ such that $a_i \mplus a_{i+1}=a_{i+1}$, for all $i\in \mathbb{N}$.
\end{definition}
Note that each element $a_i$ in the sequence for $k$-absorptiveness or countably absorptiveness are non-zero.
Furthermore,  the notions of weakly absorptive and $1$-absorptive coincide, and every self absorptive monoid is countably absorptive, but not vice versa. Finally, a monoid is weakly cancellative if and only if it is not weakly absorptive. This is immediate, since $\forall a \forall b (a \mplus b=b\rightarrow a=0)$ defines weak cancellativity and $\neg \forall a \forall b (a \mplus b=b \rightarrow a=0)$ is equivalent with $\exists a \exists b (a \mplus b= b \land a\neq 0)$ which is the definition of weakly absorptive (note that, if $a \mplus b=b \land b= 0$ then $a=0$).
These relationships are summarised in Figure \ref{fig:relationships}.

\begin{figure}
\centering
$SA \subset CA \subset  kA \subset  (k-1)A \subseteq  1A = WA = \neg WC$
\caption{Relationships between monoid properties.}
\label{fig:relationships}
\end{figure}

\section{Inclusion Dependencies}

An \emph{inclusion dependency} (IND) is an expression of the form $\sigma=\inc{R}{A_1,\dots ,A_n}{S}{B_1,\dots ,B_n}$, where $R$ and $S$ are relation names, and $A_1,\dots ,A_n\in\Att(R)$ and $B_1,\dots ,B_n\in\Att(S)$ are pairwise distinct attribute sequences (i.e., $A_i=A_j\Rightarrow i=j$ and $B_i=B_j\Rightarrow i=j$). The \emph{arity} of this IND, denoted $\ar(\sigma)$, is $n$. 
Let $\Mo$ be a positive commutative monoid and $\mD$ be a $\Mo$-database.
We say that 
$\mD$ \emph{satisfies} the IND $\inc{R}{A_1,\dots ,A_n}{S}{B_1,\dots ,B_n}$, written 
\begin{multline*}   
\mD \models \inc{R}{A_1,\dots ,A_n}{S}{B_1,\dots ,B_n},\\
\text{if } \KX^\mD[A_1,\dots,A_n](\vec{a}) \leq \KY^\mD[B_1,\dots,B_n](\vec{a}),
\end{multline*}
for all $\vec{a}\in \Const^n$. 
For an ordinary (Boolean) database $D$, we write $D \models \sigma$ if
$D \models_{\mathbb{B}} \sigma$, that is, if $\sigma$ holds under the
classical set-inclusion semantics. We extend these satisfaction relations to sets of INDs $\Sigma$ in the obvious way and write $\mD\models \Sigma$ (or $D\models \Sigma$).

One can verify that, by positivity of $\Mo$, satisfaction of inclusion
dependencies is preserved when passing from a $\Mo$-database to its
support.
\begin{proposition}\label{prop:preserved}
Let $\Mo$ be a positive commutative monoid and $\mD$ a $\Mo$-database.
If $\mD \models \sigma$ then $\supp(\mD) \models \sigma$, for any IND $\sigma$.\looseness=-1
\end{proposition}

As in the relational context, entailment between INDs in the more general monoidal context can be captured by a set of inference rules. Before introducing these rules, we briefly consider the notion of entailment in this setting.

Given a finite set of INDs $\Sigma \cup\{\tau\}$, we say that \emph{$\Sigma$ entails $\tau$ over $\Mo$-databases} (resp. \emph{over balanced $\Mo$-databases}), denoted  $\Sigma \momodels{\Mo} \tau$ (resp. $\Sigma \bmomodels{\Mo} \tau$), if for every $\Mo$-database $\mD$ (resp. every balanced $\Mo$-database $\mD$),
$\mD \models \Sigma$ implies $\mD \models \tau$. Furthermore, we write $\Sigma \models \tau$ as shorthand for
$\Sigma \models_{\mathbb{B}} \tau$, which coincides with implication of
$\tau$ by $\Sigma$ over ordinary relational databases.
As standard, we say that a provability relation $\vdash$ for INDs is \emph{sound} (\emph{complete}, resp.) with respect to a semantic entailment relation $\models$, if $\Sigma \vdash \tau$ implies $\Sigma \models \tau$ ($\Sigma \models \tau$ implies $\Sigma \vdash \tau$, resp.) for all finite sets of INDs $\Sigma \cup\{\tau\}$.

Now, in addition to the standard rules recalled in \Cref{tab:inc-rules}, the monoidal setting gives rise to further rules, shown in \Cref{tab:inc-axioms-extra}. These additional rules are not sound in general. In particular, weak symmetry is sound only when $\Mo$ is weakly cancellative, and soundness of both symmetry and balance additionally requires that we restrict to balanced $\Mo$-databases. The rest of this section is devoted to proving these claims.

\begin{table}[t]
\centering
\footnotesize
\fbox{%
\begin{minipage}{\dimexpr\columnwidth-2\fboxsep-2\fboxrule}
\textbf{Reflexivity:}\\
$\inc{R}{A_1,\dots,A_n}{R}{A_1,\dots,A_n}$, for pairwise 
distinct $A_1,\dots,A_n\in\Att(R)$.\\[0.4ex]

\vspace{-2mm}
\textbf{Transitivity:}\\
If $\inc{R}{\vec{A}}{S}{\vec{B}}$ and
$\inc{S}{\vec{B}}{T}{\vec{C}}$, then
$\inc{R}{\vec{A}}{T}{\vec{C}}$.\\[0.4ex]

\vspace{-2mm}
\textbf{Projection and permutation:}\\
If $\inc{R}{A_1,\dots,A_n}{S}{B_1,\dots,B_n}$, then
$\inc{R}{A_{i_1},\dots,A_{i_\ell}}{S}{B_{i_1},\dots,B_{i_\ell}}$,
for distinct $i_1,\dots,i_\ell \in [n]$.
\end{minipage}
}
\caption{Inference rules for inclusion dependencies.}
\label{tab:inc-rules}
\end{table}

\begin{table}[t]
\centering
\footnotesize
\setlength{\fboxsep}{4pt}
\setlength{\fboxrule}{0.4pt}
\fbox{%
\begin{minipage}{\dimexpr\columnwidth-2\fboxsep-2\fboxrule}
\textbf{Symmetry:}\\
If $\inc{R}{\vec{A}}{S}{\vec{B}}$, then $\inc{S}{\vec{B}}{R}{\vec{A}}$.\\[0.4ex]

\vspace{-2mm}
\textbf{Weak symmetry:}\\
If $\inc{R}{\vec{A}}{S}{\vec{B}}$ and $\inc{S}{\emptyset}{R}{\emptyset}$, then
$\inc{S}{\vec{B}}{R}{\vec{A}}$.\\[0.4ex]

\vspace{-2mm}
\textbf{Balance:}\\
$\inc{S}{\emptyset}{R}{\emptyset}$.
\end{minipage}}
\caption{Additional inference rules.}
\label{tab:inc-axioms-extra}
\end{table}

\begin{fact}
Symmetry is a theorem
 that follows from Weak Symmetry and Balance. 
\end{fact}
Note that $\inc{S}{\emptyset}{R}{\emptyset}$ is valid with respect to the class of balanced databases, and hence w.r.t. that class Weak Symmetry is equivalent with Symmetry. Moreover,
$\inc{R}{\emptyset}{S}{\emptyset}$
and $\inc{S}{\emptyset}{R}{\emptyset}$ together express that $R$ and $S$ are balanced.

We write $\Sigma\vdash\tau$ if $\tau$ is derivable from $\Sigma$ using the standard  three rules from \Cref{tab:inc-rules}. Similarly,  $\Sigma \wsvdash \tau$ (resp. $\Sigma \svdash \tau$) means that $\tau$ is derivable from $\Sigma$ with the standard rules extended with weak symmetry (resp. symmetry and balance) from \Cref{tab:inc-axioms-extra}. Additionally, 
  $\ordcl{\Sigma}$,  $\cl{\Sigma}$, and $\scl{\Sigma}$ stand for the closure of $\Sigma$ under $\vdash$, $\wsvdash$, and $\svdash$, respectively.

The next result shows that the usual inference rules of inclusion dependencies are sound in the monoidal context. This holds even for commutative monoids that are not positive.
\begin{proposition}\label{prop:sound}
Let $\Mo$ be a commutative monoid.
The axioms for inclusion dependencies (\Cref{tab:inc-rules}) are sound over $\Mo$-databases.  The balance axiom (\Cref{tab:inc-axioms-extra}) is sound over balanced $\Mo$-databases.
\end{proposition}
\begin{proof}
The soundness of the reflexivity and transitivity rules rely on the fact that the semantics of inclusion dependence was defined from the natural order of the monoid. Likewise, projection and permutation follows from properties of the natural order and since the definition of inclusion dependence is clearly permutation invariant. The soundness of balance axioms over balanced databases is self-evident.
\end{proof}

The next result states that weak symmetry is sound if the underlying monoid is weakly cancellative.
Unlike Proposition \ref{prop:sound}, the argument relies also on positivity of the monoid.
\begin{proposition}\label{prop:sym-sound}
Let $\Mo$ be a positive weakly cancellative commutative monoid.
Weak symmetry (\Cref{tab:inc-axioms-extra}) is sound over $\Mo$-databases. 
\end{proposition}
\begin{proof}
Let $R$ and $S$ be $\Mo$-relations from a $\Mo$-database such that $R[\vec{A}] \leq S[\vec{B}]$ and $S[\emptyset] \leq R[\emptyset]$. We will show that $S[\vec{B}] \leq R[\vec{A}]$. Let $k=\vert \vec{A}\rvert$. By assumption $R$ and $S$ are balanced ($R[\vec{A}] \leq S[\vec{B}]$ implies $R[\emptyset] \leq S[\emptyset]$), and hence
\[
\bigmplus_{\vec{a}\in \Const^k} R[\vec{A}](\vec{a}) = \bigmplus_{\vec{a}\in \Const^k} S[\vec{B}](\vec{a}).
\]
The above expressions are well-defined, since $R$ and $S$ have finite supports.
As $\leq$ is the natural order of the monoid, and $R[\vec{A}] \leq S[\vec{B}]$, there is a $c_{\vec{a}} \in K$ for each $\vec{a}\in\Const^k$  s.t.
\[
\bigmplus_{\vec{a}\in \Const^k} S[\vec{B}](\vec{a}) = \bigmplus_{\vec{a}\in \Const^k} ( R[\vec{A}](\vec{a}) \mplus c_{\vec{a}}).
\]
Since $\mplus $ is associative, we obtain that 
\[
\bigmplus_{\vec{a}\in \Const^k} ( R[\vec{A}](\vec{a}) \mplus c_{\vec{a}}) =  \bigmplus_{\vec{a}\in \Const^k}  R[\vec{A}](\vec{a}) \mplus \bigmplus_{\vec{a}\in \Const^k} c_{\vec{a}}.
\]
Putting these together yields that
\[
\bigmplus_{\vec{a}\in \Const^k} R[\vec{A}](\vec{a})  = \bigmplus_{\vec{a}\in \Const^k}  R[\vec{A}](\vec{a}) \mplus \bigmplus_{\vec{a}\in \Const^k} c_{\vec{a}}.
\]
It now follows from WC and positivity of the monoid that $c_{\vec{a}}=0$ for each $\vec{a}\in\Const^k$. Hence $S[\vec{B}] \leq R[\vec{A}]$.
\end{proof}

\begin{example}\label{ex:budget2}
Continuing the running example, we see that the rules \ref{con:exp-bud}-\ref{con:grant-bud} allow us to derive new INDs, much as in the relational context.
For instance, by projecting \ref{con:exp-bud} on $\attr{proj}$ and composing with \ref{con:bud-grant},
we obtain the derived constraint
\[
\rel{Expense}[\attr{proj}] \le \rel{Grant}[\attr{proj}],
\]
stating that total spending over the lifetime of a project cannot exceed the
grant. 

Beyond such standard derivations, monoidal reasoning also supports richer forms of inference. In the present example, $\Mo$ may be selected as the monoid of non-negative reals~$\mathbb{R}_{\geq 0}$ or natural numbers~$\mathbb{N}$. In both cases, $\Mo$ is weakly cancellative, and thus weak symmetry is sound. As a consequence, the constraint set $\{$\ref{con:exp-bud},\ref{con:bud-grant},\ref{con:grant-bud}$\}$ can be relaxed, for~\ref{con:grant-bud} can be replaced by the weaker constraint
\begin{enumerate}[
  label=(C\arabic*${}$),
  ref=C\arabic*${}$,
  start=4,
  leftmargin=*,
  labelsep=0.6em
]
\item \label{con:grant-bud-2}
$\rel{Grant}[\emptyset] \leq \rel{Budget}[\emptyset]$
\end{enumerate}
Indeed,~\ref{con:grant-bud} can be recovered from~\ref{con:grant-bud-2} and~\ref{con:bud-grant} by weak symmetry. Informally, if the total project budget does not exceed its associated grant, and the sum of all grants does not exceed the sum of all budgets across projects, then each project grant must also not exceed its corresponding budget.
\end{example}
The previous example highlights one of the advantages of inference rules: they can expose redundancies which can be computationally costly to enforce. In particular, the constraint~\ref{con:grant-bud-2} is simpler to enforce than~\ref{con:grant-bud}, from which it is obtained by projection. The former IND compares a single aggregate, while the latter involves comparisons across multiple aggregates.

The example also illustrates another benefit of the monoidal approach. INDs such as~\ref{con:exp-bud}--\ref{con:grant-bud} provide simple and natural specifications of arithmetic constraints involving aggregation. Such constraints are enforceable in the standard relational setting but not natively, as this typically requires explicit aggregation and grouping (e.g., \texttt{SUM} and \texttt{GROUP BY}) via comparatively complex queries.

Lastly, the example shows that monoidal reasoning can give rise to non-obvious and sometimes non-intuitive inferences, in contrast to the more immediate nature of classical inclusion dependency reasoning. In particular, the informal if--then statement at the end of the example conceals non-trivial algebraic reasoning that goes beyond simple projections and transitive closures.

Let us then consider the cases where weak symmetry is not sound.

\begin{proposition}
Let $\Mo$ be a positive weakly absorptive commutative monoid.
Then weak symmetry is not sound, even when restricted to the class of balanced databases.
\end{proposition}
\begin{proof}
Let $\Mo$ be WA monoid and $a,b \in K$ non-zero elements such that $a\mplus b = b$. \Cref{tbl:WA} depicts a balanced $\Mo$-database that is a counterexample for weak symmetry; $\rel{Warehouse}[\attr{product}] \leq \rel{Orders}[\attr{product}]$ and $\rel{Orders}[\emptyset] \leq \rel{Warehouse}[\emptyset]$ hold, but $\rel{Orders}[\attr{product}] \leq \rel{Warehouse}[\attr{product}]$ does not hold.
\end{proof}
\begin{table}
    \footnotesize
\centering
\begin{tabular}{@{}l: l@{}}
\multicolumn{2}{l}{$\rel{Warehouse}$}\\
\toprule
\attr{product} & \# \\
\midrule
potato& $b$ \\
yam & $0$ \\
\bottomrule
\end{tabular}
\quad
\begin{tabular}{@{}l: l@{}}
\multicolumn{2}{l}{$\rel{Orders}$}\\
\toprule
\attr{product} & \# \\
\midrule
potato& $b$ \\
yam & $a$ \\
\bottomrule
\end{tabular}
\caption{Example of a balanced $\Mo$-database, where $a\mplus b=b$.}
\label{tbl:WA}
\end{table}

Having analysed the soundness of our rules, we now turn to their completeness. To this end, we first consider an algorithm that is commonly used to reason about the implication problem for database dependencies: \emph{the chase}.

\section{$\mplus$-Chase}
Next, we define a variant of the chase that will be used in the completeness proofs and may be of independent interest.
Throughout the section, we assume that $\Mo$ is a non-trivial positive commutative monoid that is weakly cancellative and equipped with a {total natural order}. In particular, this means that
the natural order is antisymmetric.
Moreover, if $a \not\leq b$, there is a (unique\footnote{Assuming $b+c=a$ and $b+c'=a$, 
by totality and symmetry we may assume w.l.o.g. that $c \leq c'$. Then there is $d$ such that
\(
c + d = c',
\)
and hence
\(
a = b + c' = b + (c + d) = (b + c) + d = a + d.
\)
By weak cancellativity, $a + d = a$ implies $d = 0$, and thus $c = c'$.
}) element $c$ such that $b + c = a$; we denote this element  by $a - b$.
Later, we show how the results of this section can be utilised even when the underlying monoid is equipped with a non-total natural order.

Let $\mD$ be a $\Mo$-database over a schema $\sch$, and let  $\Sigma$ be a set of INDs. 
Let $*$ be a constant that is not in $\adom(\mD)$.
We consider the following chase rule, adapted from the relational setting of Casanova, Fagin, and Papadimitriou~\shortcite{CasanovaFP84}.

\vspace{2mm}
\noindent \textbf{$\oplus$-rule.} Let $\sigma=\inc{R}{\vec{A}}{S}{\vec{B}}\in \Sigma$, and 
suppose there exists a sequence of constants $\vec{a}$ such that $ \marg{R^{\mD}}{\vec{A}}(\vec{a}) \not\leq \marg{S^{\mD}}{\vec{B}}(\vec{a})$.
Then we say that the $\oplus$-rule is \emph{applicable to $\mD$ w.r.t. $\sigma,\vec{a}$}.
Let $t'\colon \Att(S)\to \Const$ be a tuple such that $t'(\vec{B})=\vec{a}$, and $t'(C)=*$ for each attribute $C\in \Att(S)$ that is not listed in $\vec{B}$. 
Let $\mD'$ be the database obtained from $\mD$ by leaving all $\Mo$-relations 
$U \neq S$ unchanged, and setting %
\begin{align*}
S^{{\mD}'}(t') &\dfn
  S^{\mD}(t') \mplus  (  \marg{R^{\mD}}{\vec{A}}(\vec{a}) -  \marg{S^{\mD}}{\vec{B}}(\vec{a})), \text{ and }\\
S^{\mD'}(t) &\dfn S^{\mD}(t), \text{ for all } t \neq t'.
\end{align*}
Then $\mD'$ is the \emph{result of applying the $\oplus$-rule to $\mD$ w.r.t. $\sigma,\vec{a}$}, and the corresponding chase step is denoted $\chasenext{\mD}{\mD'}{\sigma,\vec{a}}$.
Moreover, we say that this chase step \emph{increments $t'$}. We observe that 
the chase step is monotone, i.e., $U^{\mD}(t) \leq U^{\mD'}(t)$ for all $U\in \sch$ and $t\in \tup(U)$. Moreover, we obtain the following equality:
\begin{align}
 \marg{S^{\mD'}}{\vec{B}}(\vec{a}) &= \marg{S^{\mD}}{\vec{B}}(\vec{a}) \mplus (  \marg{R^{\mD}}{\vec{A}}(\vec{a}) -  \marg{S^{\mD}}{\vec{B}}(\vec{a}))\nonumber\\
 &=  \marg{R^{\mD}}{\vec{A}}(\vec{a}).\label{eq:equals}
\end{align}

For a tuple $t$ that appears in the context of the chase, the \emph{degree} of $t$ refers to the number attributes it maps to a value different from $*$, that is, 
\begin{equation}\label{eq:degree}
    \deg(t)\coloneqq |\{A\in \Dom(t)\mid t(A)\neq *\}|.
\end{equation}

\vspace{1mm}
\noindent \textbf{$\oplus$-chase sequence.}
A \emph{$\oplus$-chase sequence (from $\mD$ by $\Sigma$)} is any sequence of $\Mo$-databases $(\mD_i)_{i\geq 0}$, where $\mD_0=\mD$,
\begin{enumerate}
    \item\label{it:twoplus} for each $i\geq 0$, $\chasenext{\mD_i}{\mD_{i+1}}{\sigma_i,\vec{a}_i}$ for some $\sigma_i\in \Sigma$ and $\vec{a}\in \Const^{\ar(\sigma_i)}$,
    \item\label{it:three} for each $i\geq 0$, %
    $\sigma=\inc{R}{\vec{A}}{S}{\vec{B}}\in \Sigma$, and $\vec{a}\in \Const^{\ar(\sigma)}$,  there exists $j> i$, 
    $ \marg{R^{\mD_i}}{\vec{A}}(\vec{a}) \leq  \marg{S^{\mD_j}}{\vec{B}}(\vec{a})$. 
\end{enumerate}
Informally, the last condition states that every IND must eventually be
``closed.''  
It is not hard to see that one can always construct a chase sequence satisfying the above conditions. 
\begin{lemma}\label{lem:exists}
     Let $\Mo$ be a non-trivial positive commutative monoid that is weakly cancellative and equipped with a {total natural order}.
   Let $\mD$ be a $\Mo$-database and
   $\Sigma$ a finite set of INDs. %
   Then, there exists a $\oplus$-chase sequence from $\mD$ by $\Sigma$.
\end{lemma}
\begin{proof}
    Choose some ordering of pairs $(\sigma,\vec{a})$, where $\sigma$ is an IND of the form $\inc{R}{\vec{A}}{S}{\vec{B}}\in \Sigma$, and $\vec{a}\in (\adom(\mD)\cup\{*\})^{|\vec{A}|}$. During the chase, each chase step updates a variable $(\sigma,\vec{a})$. This variable is set to be the next one (assuming the last element is followed by the first element) with respect to which the $\oplus$-rule
is applicable to the current $\Mo$-database.
Using~\eqref{eq:equals}, it is straightforward to verify that
\cref{it:three} is satisfied by the resulting chase sequence.
\end{proof}

Given a finite $\oplus$-chase sequence
$(\mD_0,\dots,\mD_n)$, we call the final $\Mo$-database $\mD_n$ the
\emph{result of the $\oplus$-chase of $\mD$ by $\Sigma$}.
If every $\oplus$-chase sequence is finite, we say that the
\emph{$\oplus$-chase of $\mD$ by $\Sigma$ terminates}.
In this case, we write $\chaseplus{\mD}{\Sigma}$ to denote
{an arbitrary result} of the $\oplus$-chase of $\mD$ by $\Sigma$.
As the next example demonstrates, the $\mplus$-chase need not terminate, and even when it does,
it may fail to produce a unique result.
\begin{example}
Let $t$ be a tuple over $\{A,B,C\}$ such that $t(A,B,C)=(a,b,c)$.
Consider an $\mathbb{N}$-database $\mD=\{R\}$, where $R(t)=1$, and $R(t')=0$ for $t'\neq t$. 
Let $\Sigma=\{\inc{R}{B,C}{R}{A,B}\}$. Then the $\mplus$-chase sequence from $\mD$ by $\Sigma$ is infinite.
This is illustrated by the relation on the left-hand side of \Cref{tab:pluschase}, in which one row represents one chase step
which increments the multiplicity of this row by $1$. 

On the other hand, $\Sigma$ is equivalent with its symmetric closure $\Sigma'=\{\inc{R}{B,C}{R}{A,B},\inc{R}{A,B}{R}{B,C}\}$. Now, the $\mplus$-chase sequence from $\mD$ by $\Sigma'$ is finite, as illustrated by the relation on the right-hand side of \Cref{tab:pluschase}. Note that the colored tuple in this relation can either be kept or removed, resulting in two different $\mplus$-chase sequences. These two sequences illustrate that the result of the $\mplus$-chase, if it exists, is not necessarily unique.
    \begin{table}
    $$
    \footnotesize
\begin{array}{ccc}
\begin{array}[t]{ccc}
  A & B & C \\
\toprule
   a & b & c \\
   b & c & * \\
   c & * & * \\
   * & * & * \\
   * & * & * \\
   \vdots & \vdots & \vdots 
\end{array}
&
\quad
&
\begin{array}[t]{cccc}
 A & B & C \\
\toprule
   a & b & c \\
   b & c & * \\
   c & * & * \\
  \cellcolor{gray!20} * & \cellcolor{gray!20}* & \cellcolor{gray!20}* \\ %
   * & a & b \\
   * & c & * \\
   c & * & * \\
   * & * & c \\
   * & * & a 
\end{array}
\end{array}
$$
\caption{$\mplus$-chase example. Left: chasing the first row by $\inc{R}{B,C}{R}{A,B}$ does not terminate. Right: chasing the first row by $\inc{R}{B,C}{R}{A,B}$ and $\inc{R}{A,B}{R}{B,C}$ terminates. \label{tab:pluschase}}
\end{table}

\end{example}
The previous example demonstrates that the $\mplus$-chase terminates if we close
the dependency set by symmetry. Note that the
example is confined to the uni-relational setting, where symmetry is sound. In the multirelational setting, however, symmetry is not
 sound and therefore cannot be assumed freely. The following result shows that
termination can nevertheless be guaranteed by closing the set of inclusion
dependencies by weak symmetry, together with the standard axioms.
The proof, deferred to the supplementary material, adapts
\cite[Theorem~33]{HannulaV22} in three directions: it
moves from the uni-relational to the multirelational setting, from the
non-negative reals to an arbitrary monoidal context, and makes explicit
the exact monoid axioms and closure properties on which the argument depends.
Specifically, the results that follow require only totality of the natural order, while the underlying monoid need not be positive or weakly cancellative. 

\begin{restatable}{lemma}{terminates}
\label{lem:terminates}
    Let $\Mo$ be a non-trivial positive commutative monoid that is weakly cancellative and equipped with a {total natural order}.
   Let $\mD$ be a $\Mo$-database over a schema $\sch$ and
   $\Sigma$ a finite set of INDs that is closed under $\wsvdash$. Then the $\oplus$-chase of $\mD$ by $\Sigma$ terminates, i.e., $\chaseplus{D}{\cl{\Sigma}}$ exists.
\end{restatable}

We thus obtain the following results using Lemmas \ref{lem:exists} and \ref{lem:terminates}, and the $\mplus$-chase construction.
\begin{theorem}\label{thm:plusterminates}
        Let $\Mo$ be a non-trivial positive commutative monoid that is weakly cancellative and equipped with a {total natural order}.
   Let $\mD$ be a $\Mo$-database, and
   $\Sigma$ a finite set of INDs. Then $\chaseplus{\mD}{\cl{\Sigma}}$ is a $\Mo$-database satisfying $\Sigma$.
\end{theorem}

\begin{corollary}\label{cor:Nchase}
   Let $\mD$ be an $\mathbb{N}$-database and
   $\Sigma$ a finite set of INDs. $\chaseplus{\mD}{\cl{\Sigma}}$ is an $\mathbb{N}$-database satisfying $\Sigma$.\looseness=-1
\end{corollary}

\section{Weakly Cancellative Monoids}
We now turn to weakly cancellative monoids. Using the results established in the previous section, we show that the axioms of inclusion dependencies extended with weak symmetry are sound and complete over $\Mo$-databases, for any non-trivial positive commutative weakly cancellative monoid $\Mo$.
Additionally, we obtain that the $\mplus$-chase defined in the previous section characterises $\Mo$-implication in this setting.

We start by presenting two lemmas; the first is proven in the supplementary material.
A monoid homomorphism  $f\colon \Mo \to \Mo'$ is called an \emph{order embedding} if it is injective, and $a\leq b$ iff $f(a)\leq f(b)$ (where $\leq$ denotes the natural orders). If such an $f$ exists, we say that $\Mo$ \emph{order embeds into} $\Mo'$.
\begin{restatable}{lemma}{WCEmbed}\label{lem:wcembed}
    Let $\Mo$ be a non-trivial weakly cancellative positive commutative monoid. Then $\mathbb{N}$ order embeds into $\Mo$.
\end{restatable}

\begin{lemma}\label{lem:Membed}
Let $\Mo$ and $\Mo'$ be monoids s.t.  $\Mo$ order embeds into $\Mo'$ and $\Sigma\cup\{\sigma\}$ a set of INDs. Then $\Sigma \models_{\Mo'} \sigma \Rightarrow  \Sigma \models_{\Mo} \sigma$.
\end{lemma}
\begin{proof}
For a $\Mo$-database $\mD$, define an $\Mo'$-database $\mD'$ via $R^{\mD'}(t)\dfn f(R^{\mD}(t))$. By Lemma \ref{lem:wcembed} it is straightforward to check that $\mD \models \tau$ if and only if $\mD'\models \tau$, for any IND $\tau$. From this, the lemma statement follows.
\end{proof}

With these lemmas in place, we can now establish soundness and
completeness via a canonical $\mplus$-chase construction. To this end, we
associate an IND $\tau$ with a canonical starting database $\mD_\tau$.
Then, we show that semantic entailment of $\tau$ over $\Mo$, its syntactic
derivability, and its satisfaction in the $\mplus$-chase of $\mD_\tau$ are
equivalent, provided that the chase is performed over the derivable closure
of the assumption set. From now on, we assume that the natural numbers are included in the constants.

\vspace{2mm}
\noindent\textbf{Canonical $\mathbb{N}$-database $\mD_\tau$.}
For 
\begin{equation}\label{eq:chasetau}
    \tau=\inc{R}{A_1, \dots ,A_n}{S}{B_1, \dots ,B_n},
    \end{equation}
define $t_\tau\colon\Att(R)\to \Const$ such that $t_\tau(A_i)=i$, for $i\in [n]$, and otherwise  $t_\tau(C)=*$.
Then, define $\mD_\tau$ as the $\mathbb{N}$-database $\{R^D\}$ such that $R^D(t_\tau)=1$, and $R^D(t')=0$ for $t'\neq t$.

\begin{theorem}\label{prop:implicit}
    Let $\Mo$ be a non-trivial weakly cancellative positive commutative monoid.
    Let $\Sigma\cup \{\tau\}$ be a finite set of INDs. Then the following statements are equivalent:\\
    \begin{enumerate*}
        \item\label{it:second} $\Sigma\wsvdash \tau$;\hspace{5mm}
        \item\label{it:first} $\Sigma\momodels{\Mo} \tau$;\hspace{5mm}
        \item\label{it:third} $\chaseplus{\mD_\tau}{\cl{\Sigma}}\models \tau$.
    \end{enumerate*}
\end{theorem}
\begin{proof}
    (\ref{it:second} $\Rightarrow$ \ref{it:first}) Follows by Propositions \ref{prop:sound} and \ref{prop:sym-sound}.
    (\ref{it:first} $\Rightarrow$ \ref{it:third}) By Corollary \ref{cor:Nchase}, $\chaseplus{\mD_\tau}{\cl{\Sigma}}$ is a $\mathbb{N}$-database satisfying $\tau$. Since by Lemmas \ref{lem:wcembed} and \ref{lem:Membed}, we have $\Sigma\momodels{\mathbb{N}} \tau$, it follows that 
     $\chaseplus{\mD_\tau}{\cl{\Sigma}}\models \tau$.

    (\ref{it:third} $\Rightarrow$ \ref{it:second}) We only sketch this direction, as it is a simple adaptation from \cite[Theorem 3.1]{CasanovaFP84}. Since $\Sigma \wsvdash \cl{\Sigma}$, it suffices to show that $\chaseplus{\mD_\tau}{\Sigma'}\models \tau'$ implies $\Sigma'\wsvdash \tau'$ for any set of INDs $\Sigma'\cup\{\tau'\}$. For this, assuming $\tau$ takes the form \eqref{eq:chasetau}, it suffices to prove inductively on $i\geq 0$ that $\Sigma' \wsvdash\inc{R}{A_{i_1},\dots ,A_{i_k}}{S}{C_1,\dots ,C_k}$ whenever we find a $T$-tuple $t$, for some relation name $T$, and attributes $C_1, \dots ,C_n \in \Att(T)$, s.t. $T^{\mD_i}(t)\neq 0$ and $t(C_1, \dots ,C_k)=(i_1,\dots ,i_k)\in [n]^k$. In fact, we can show $\Sigma' \vdash\inc{R}{A_{i_1},\dots ,A_{i_k}}{S}{C_1,\dots ,C_k}$, as only the standard inference rules are needed in the induction. Then, given $n\in \mathbb{N}$ s.t. $D_n=\chaseplus{\mD_\tau}{\Sigma'}$, by assumption we find an $S$-tuple $t$ s.t. $S^{\mD_n}(t)\neq 0$ and $t(B_1, \dots ,B_n)=(1, \dots ,n)$. This entails \eqref{it:second} by the induction argument.\looseness=-1
\end{proof}

As a result we obtain the following theorem.
\begin{theorem}\label{thm:wc}
    Let $\Mo$ be a weakly cancellative positive commutative monoid.
    The axioms of inclusion dependencies extended with weak symmetry are sound and complete over $\Mo$-databases.
\end{theorem}
For balanced $\Mo$-databases, similar axiomatic characterization follows as a straightforward consequence.
\begin{corollary}\label{cor:wc}
    Let $\Mo$ be a weakly cancellative positive commutative monoid.
    The axioms of inclusion dependencies extended with balance axioms and weak symmetry are sound and complete over balanced $\Mo$-databases.
\end{corollary}
\begin{proof}
Soundness follows from Propositions \ref{prop:sound} and \ref{prop:sym-sound}.
For completeness, suppose $\Sigma\bmomodels{\Mo}\sigma$. 
Then $\Sigma^* \models_{\Mo}\sigma$, where $\Sigma^*$ extends $\Sigma$ by all balance axioms $\inc{S}{\emptyset}{R}{\emptyset}$,  where $R$ and $S$ are relation names appearing in $\Sigma$ or $\sigma$. %
It follows by \Cref{thm:wc} that $\sigma$ is derivable from $\Sigma^*$ using the axioms of inclusion dependencies extended with weak symmetry. From this the statement of the corollary follows.
\end{proof}

\section{Weakly Absorptive Monoids}
Next, we show that the standard axioms of INDs are sound and complete over $\Mo$-databases, for weakly absorptive positive commutative monoids. 
We distinguish two cases depending on whether $\Mo$ is countably absorptive.

\subsection{Countably Absorptive Monoids}
As a warm-up, let us consider self-absorptive monoids as a special case of countably absorptive ones. In this case, soundness and completeness follow readily from earlier results.
\begin{theorem}\label{thm:sa}
    Let $\Mo$ be a self-absorptive positive and commutative monoid. The axioms of inclusion dependencies are sound and complete for the implication problem over $\Mo$.
\end{theorem}
\begin{proof}
Soundness follows by Prop.~\ref{prop:sound}. For completeness, assume that $\Sigma \models_{\Mo} \sigma$. Let $b$ some non-zero self absorptive element of $\Mo$. By positivity of $\Mo$ we obtain that the Boolean monoid $\mathbb{B}$ order embeds into $\Mo$ via the mapping $0_\B \mapsto 0_\Mo$ and $1_\B \mapsto b$.
By Lemma \ref{lem:Membed}, $\Sigma \models_\B\sigma$. Since, the axioms are complete in the case of the Boolean monoid \cite{CasanovaFP84}, we get that $\Sigma \vdash\sigma$.
\end{proof}

For countably absorptive monoids in general, more work is required. In what follows, we use the ordinary chase procedure for relational inclusion dependencies to construct a stratified counterexample along the countable sequence of absorptive elements. To this end, we briefly recall the notion of the chase in this setting.

\vspace{2mm}
\noindent \textbf{Chase for INDs.}
Let $D$ be a database and $\Sigma$ a set of INDs. Again, we assume $*$ is a constant which does not belong to the active domain of $D$. 
Consider the following chase rule from
\cite{CasanovaFP84}:

\noindent
\textbf{Rule (*).} Let $t\in R^D$, and let $\inc{R}{\vec{A}}{S}{\vec{B}}\in \Sigma$. Then extend $S^D$ with the tuple $t'$ (if $t'$ is not already in $S^D$) such that $t'(Y_i)=t(X_i)$, for $i\in [n]$, and $t'(Y)=*$ otherwise. 

Given a database $D$ and a set of INDs $\Sigma$, denote by $\chase{D}{\Sigma}$ the closure of $D$ under $\Sigma$ by using the above rule. Note that 
$\chase{D}{\Sigma}$ is finite, uniquely determined, and, by construction, satisfies $\Sigma$.
Furthermore, given an IND $\tau=\inc{R}{X_1, \dots ,X_n}{S}{Y_1, \dots ,Y_n}$, denote by $D_\tau\dfn\{R^D\}$ the canonical starting database, where $R^D$ consists of one $R$-tuple $t$ such that $t(X_i)=i$, for $i\in [n]$, and otherwise  $t(X)=*$.
The following result is implicit in \cite[Proof of Theorem 3.1]{CasanovaFP84}; alternatively, see \cite[Chapter 11]{ABLMP21}.
\begin{theorem}\label{thm:known}
    Let $\Sigma$ be a finite set of INDs, and let $\tau$ be an IND. Then the following are equivalent:\\
    \begin{enumerate*}
        \item $\Sigma\vdash \tau$;\hspace{5mm}
        \item $\Sigma\models \tau$;\hspace{5mm}
        \item $\chase{D_\tau}{\Sigma}\models \tau$.
    \end{enumerate*}

\end{theorem}

 The previous characterization can now be used to obtain soundness and completeness over countably absorptive positive commutative monoids
\begin{theorem}\label{prop:cwa}
    Let $\Mo$ be a countably absorptive positive commutative monoid. The axioms of inclusion dependencies are sound and complete for the implication problem over $\Mo$.
\end{theorem}
\begin{proof}
    As soundness follows by Proposition \ref{prop:sound}, we focus on completeness. 
    Suppose $\Sigma \not\vdash \tau$. To show $\Sigma \not\models_{\Mo} \tau$, we create a counterexample model witnessing this.
    By \Cref{thm:known},  $D\dfn \chase{D_\tau}{\Sigma}$ is a relational database satisfying $\Sigma$ and violating $\tau$. In what follows,
    we turn $D$ into a $\Mo$-database.

    Recall the definition of the {degree} of a tuple $t$ from \eqref{eq:degree}. %
    Suppose $\tau$ is of the form $\inc{R}{A_1, \dots ,A_n}{S}{B_1,\dots ,B_n}$. Let $(a_i)_{i\geq 0}$ be a sequence of non-zero values from $\Mo$ witnessing the countably absorptive property. 
    We now define a $\Mo$-database $\mD$ from $D$ as follows. For each relation $T^D$ of $D$, we set a $\Mo$-relation $T^{\mD}$ into $\mD$ such that
    for all $t \in \tup(T)$,
    \begin{equation}\label{eq:Tdef}
    T^{\mD}(t) \dfn
    \begin{cases}
        a_i &\text{ if $t\in T^D$ and $i = n - \deg(t)$,}\\
        0_{\Mo} &\text{ otherwise.}
    \end{cases}
    \end{equation}
    In particular, the initial tuple of the chase has degree $n$ and is consequently mapped to $a_0$. 

    Considering $\tau$, note that $R^D$ includes a tuple $t$ (i.e., the initial one) such that $t(A_1, \dots ,A_n)=(1, \dots ,n)$, but $S^D$ does not include any tuple $t'$ such that $t'(B_1, \dots ,B_n)=(1, \dots ,n)$. Hence 
    $\marg{R^{\mD}}{A_1,\dots ,A_n}(1, \dots ,n)= a_i$, for some $i\geq 0$, while $\marg{S^{\mD}}{B_1,\dots ,B_n}(1, \dots ,n)= 0_{\Mo}$. This entails $\mD \not\models \tau$ by positivity of $\Mo$.

    Then, let $\sigma=\inc{T}{\vec{C}}{U}{\vec{D}}\in \Sigma$, and let $\vec{c}\in \Const^{|\vec{C}|}$ be such that
    $\marg{T^{\mD}}{\vec{C}}(\vec{c})\neq 0$. By construction, we find $0\leq i\leq n$ such that 
    $\marg{T^{\mD}}{\vec{C}}(\vec{c})= a_i$. Hence $T^D$ contains a tuple $t$ with $\deg(t)=n-i$
    and $t(\vec{C})=\vec{c}$. 
    Since $D\models \Sigma$, $U^D$ contains a tuple $t'$ with $t'(\vec{D})=\vec{c}$ and
    \begin{equation}\label{eq:chase}
    \deg(t')=\deg(\marg{t'}{\vec{D}})= \deg(\marg{t}{\vec{C}})\leq n-i.
    \end{equation}
    Namely, $t'$ is the result of applying the chase rule to $\sigma$ and $t$. Then, in particular, the first equality in \eqref{eq:chase} follows as $t'$ maps attributes that are not listed in $\vec{D}$ to $*$. 
    
    We now obtain by \eqref{eq:Tdef} and \eqref{eq:chase} that $U^{\mD}(t')= a_j$ for $j\geq i$. Hence 
    \[
    \marg{U^{\mD}}{\vec{D}}(\vec{c})\geq U^{\mD}(t') =a_j \geq a_i,
    \]
    showing that $\mD\models \sigma$. Since $\sigma$ was chosen arbitrarily from $\Sigma$, we have $\mD \models \Sigma$.
    This concludes the proof of $\Sigma\not\models_{\Mo}\tau$, and thus completeness follows.
\end{proof}

\subsection{Not Countably Absorptive Monoids}
Let us then turn to the case of weakly absorptive monoids that are not countably absorptive.
Given a monoid $\Mo$ and its non-zero element $b$, we write $\Mo_b$ for the submonoid 
of $\Mo$ generated by $b$. That is,  $\Mo_b$ consist of elements of the form
\[
kb \defeq \overbrace{b \mplus \cdots \mplus b}^{k\text{ times}},
\]
where $k \in \mathbb{N}$, with the convention that $0b$ denotes the neutral element $0$.
We say that $\Mo_b$ is \emph{eventually periodic} if there exists $m,\ell\in\N_{>0}$ such that $mb=(m+\ell)b$.
Note that, in a positive monoid, this equality cannot hold with $m=0$. The following lemma classifies weakly but not countably absorptive monoids into two categories: those with a total natural order and those that are eventually periodic.
The proof, placed in the supplementary material, constructs increasing chains of elements illustrated in
\Cref{fig:chain}: the top diagram corresponds to \cref{it:two}, and the
bottom diagram to \cref{it:two-half}. Intuitively, the first case behaves like an initial segment of the natural
numbers under $\max$-addition, followed by the positive natural numbers
under ordinary addition, whereas the second case eventually closes this additive tail into a cycle,
in which addition resembles that of a finite field.
\begin{restatable}{lemma}{WeakAbsorption}\label{lem:wa}
	Let $\Mo$ be a positive commutative monoid that is weakly absorptive but not countably absorptive. Then there exist $a,b\in\Mo \setminus \{0\}$ such that
	\begin{enumerate}
        \item\label{it:one} %
        $a \mplus b = b$;
        \item\label{it:one-half}  
        $b \mplus c \neq c$ for all $c\in \Mo$; 
        \item\label{it:ntwcpc} $\Mo_b$ is non-trivial, positive, commutative; 
		\item\label{it:two} %
      if the natural order of $\Mo_b$ is antisymmetric, then $\Mo_b$ is weakly cancellative and has a total natural order; 
      		\item\label{it:two-half} %
      if the natural order of $\Mo_b$ is not antisymmetric, then $\Mo_b$ is eventually periodic.
	\end{enumerate}
\end{restatable}

Before proceeding to the main result of this section, we consider the following helping lemma, proven in the supplementary material.

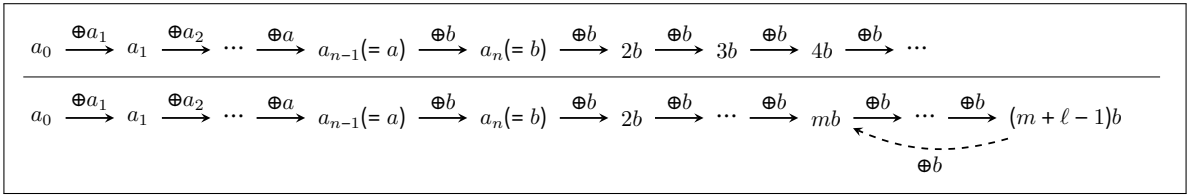
\begin{figure*}[t]
\centering
\setlength{\fboxsep}{5pt}
\setlength{\fboxrule}{0.4pt}

\fbox{%
\scalebox{0.85}{
\begin{minipage}{\textwidth}
\raggedright

\begin{tikzpicture}[
  >=stealth,
  thick,
  node distance=.9cm,
  shorten >=2pt,
  shorten <=2pt
]
\node (a0) {$a_0$};
\node (a1) [right=of a0] {$a_1$};
\node (a2) [right=of a1] {$\cdots$};
\node (an1) [right=of a2] {$a_{n-1}(=a)$};
\node (an) [right=of an1] {$a_n(=b)$};

\node (b2) [right=of an] {$2b$};
\node (b3) [right=of b2] {$3b$};
\node (b4) [right=of b3] {$4b$};
\node (dots) [right=of b4] {$\cdots$};

\node (bd)   [right=of dots] {$\phantom{\cdots}$};
\node (bm)   [right=of bd]   {$\phantom{mb}$};
\node (bd2)  [right=of bm]   {$\phantom{\cdots}$};
\node (bml)  [right=of bd2]  {$\phantom{(m+\ell-1)b}$};

\draw[->] (a0) -- node[above] {$\mplus a_1$} (a1);
\draw[->] (a1) -- node[above] {$\mplus a_2$} (a2);
\draw[->] (a2) -- node[above] {$\mplus a$} (an1);
\draw[->] (an1) -- node[above] {$\mplus b$} (an);

\draw[->] (an) -- node[above] {$\mplus b$} (b2);
\draw[->] (b2) -- node[above] {$\mplus b$} (b3);
\draw[->] (b3) -- node[above] {$\mplus b$} (b4);
\draw[->] (b4) -- node[above] {$\mplus b$} (dots);
\end{tikzpicture}

\vspace{0.8ex}
\hrule
\vspace{0.8ex}

\begin{tikzpicture}[
  >=stealth,
  thick,
  node distance=.9cm,
  shorten >=2pt,
  shorten <=2pt
]
\node (a0) {$a_0$};
\node (a1) [right=of a0] {$a_1$};
\node (a2) [right=of a1] {$\cdots$};
\node (an1) [right=of a2] {$a_{n-1}(=a)$};
\node (an) [right=of an1] {$a_n(=b)$};

\node (b2) [right=of an] {$2b$};
\node (bd) [right=of b2] {$\cdots$};
\node (bm) [right=of bd] {$mb$};
\node (bd2) [right=of bm] {$\cdots$};
\node (bml) [right=of bd2] {$(m+\ell-1)b$};

\draw[->] (a0) -- node[above] {$\mplus a_1$} (a1);
\draw[->] (a1) -- node[above] {$\mplus a_2$} (a2);
\draw[->] (a2) -- node[above] {$\mplus a$} (an1);
\draw[->] (an1) -- node[above] {$\mplus b$} (an);

\draw[->] (an) -- node[above] {$\mplus b$} (b2);
\draw[->] (b2) -- node[above] {$\mplus b$} (bd);
\draw[->] (bd) -- node[above] {$\mplus b$} (bm);
\draw[->] (bm) -- node[above] {$\mplus b$} (bd2);
\draw[->] (bd2) -- node[above] {$\mplus b$} (bml);

\draw[->,  dashed, bend left=20] (bml) to node[below] {$\mplus b$} (bm);
\end{tikzpicture}

\end{minipage}
}
}

\caption{Chains generated in WA\&$\neg$CA monoids. 
Top: an unbounded $b$-tail (when natural order is antisymmetric over multiplicities of $b$). 
Bottom: an eventually periodic $b$-tail, with $(m+\ell)b=mb$ (when natural order is not antisymmetric over multiplicities of $b$).}
\label{fig:chain}
\end{figure*}

\begin{restatable}{lemma}{SuppPreserved}\label{lem:supppreserved}
         Let $\Mo$ be a commutative monoid with a total natural order, $\Sigma$ a set of INDs that is closed under $\wsvdash$, and $\mD$ a $\Mo$-database satisfying $\Sigma$. If $\chase{\supp(\mD)}{\Sigma})=\supp(\mD)$, then $\supp(\chaseplus{\mD}{\Sigma})=\supp(\mD)$.
\end{restatable}

We can now establish soundness and completeness for weakly absorptive monoids that are
not countably absorptive. The completeness proof splits into two
cases, corresponding to the two possible $b$-tails identified in Lemma \ref{lem:wa}.
\begin{theorem}\label{prop:fwa}
    Let $\Mo$ be a positive commutative monoid that is weakly absorptive but not countably absorptive. Then the axioms of inclusion dependencies are sound and complete for the implication problem over $\Mo$.
\end{theorem}
\begin{proof}
As soundness follows by Proposition \ref{prop:sound}, we consider only completeness. 
Let $\Sigma \cup\{\tau\}$ be a finite set of INDs. 
Assuming $\Sigma\not\vdash \tau$, we prove $\Sigma\not\models_{\Mo} \tau$.
Consider the database $D\defeq \chase{D_\tau}{\ordcl{\Sigma}}$, obtained via the ordinary chase from the canonical starting database $D_\tau$, with respect to the  closure $\ordcl{\Sigma}$ of $\Sigma$ under the axioms of INDs.
By Theorem \ref{thm:known}, we have
\begin{equation}\label{eq:counter}
    D\models \ordcl{\Sigma} \text{ and }D \not\models \tau.
\end{equation}
By construction, for any tuple $t\in R^{D}$, we have $\deg(t)\leq \ar(\sigma)$.
We partition $D$ into two databases $D_0$ and $D_1$, where
$R^{D_0} \dfn \{t\in R^D\mid \deg(t)=\ar(\tau)\}$ and
$R^{D_1} \dfn \{t\in R^D\mid \deg(t)<\ar(\tau)\}$, 
for each $R \in \sch$ where $\sch$ is the schema of $D$.
Furthermore, we set $D'_1 \coloneqq \chase{D_1}{\cl{\Sigma}}$.
     
By Lemma \ref{lem:wa}, there exist
$a,b\in\Mo \setminus \{0\}$ s.t. $a\mplus b = b$, and $b \mplus c\neq c$ for all $c\in \Mo$.
We divide the proof into two cases.

\vspace{2mm}
\noindent\textbf{Case 1: the natural order of $\Mo_b$ is antisymmetric.}
Define $\Mo$-databases $\mD_0$ and $\mD_1$ by setting, for each $R \in \sch$,
$R^{\mD_0}(t)=a$ for all $t\in R^{D_0}$ and $R^{\mD_1}(t)=b$ for all
$t\in R^{D'_1}$, and mapping all other tuples to $0$.

By Lemma \ref{lem:wa}, the submonoid $\Mo_b$ is non-trivial, positive, commutative, weakly cancellative, and associated with a total natural order.
Define
\[
\mD \defeq \mD_0\mplus \chaseplus{\mD_1}{\cl{\Sigma}}.
\]
 By Theorem \ref{thm:plusterminates}, this $\Mo$-database is well-defined.
We claim that $\mD$ satisfies $\Sigma$ but violates $\tau$.

    Since $D'_1=\supp(\mD_1)$ and $\chase{D'_1}{\cl{\Sigma}}=D'_1$ (by construction),
     Proposition \ref{prop:supp}, and Lemma \ref{lem:supppreserved} imply that
    \begin{align*}
        &\supp(\mD_0\mplus \chaseplus{\mD_1}{\cl{\Sigma}}) \\
        = &\supp(\mD_0)\cup \supp(\chaseplus{\mD_1}{\cl{\Sigma}}) \\
        = &\supp(\mD_0)\cup \supp(\mD_1)=\supp(\mD_0\mplus \mD_1)  \\
        = & D.
    \end{align*}
    The last equality follows by positivity of $\Mo_b$.
    Consequently, by \eqref{eq:counter} and Proposition \ref{prop:preserved} we obtain $\mD \not\models \tau$,

    To show that $\mD \models \Sigma$, consider an arbitrary inclusion dependency $\inc{R}{\vec{A}}{S}{\vec{B}} \in \Sigma$. We need to show that, given $\vec{a}\in \Const^{|\vec{A}|}$, 
    \begin{equation}\label{eq:ineqind}
        \marg{R^{\mD}}{\vec{A}}(\vec{a})\leq \marg{S^{\mD}}{\vec{B}}(\vec{a}).
        \end{equation}
        
    Suppose first that $\deg(\vec{A}\mapsto \vec{a})=\ar(\tau)$. Since at most one tuple created by the chase can project to $(\vec{A}\mapsto \vec{a})$, we obtain 
    \[
    \marg{R^{\mD}}{\vec{A}}(\vec{a})= \marg{R^{\mD_0}}{\vec{A}}(\vec{a}) \in \{0,a\}.
    \]
    If $\marg{R^{\mD_0}}{\vec{A}}(\vec{a})=0$, \eqref{eq:ineqind} readily follows.
    If $\marg{R^{\mD_0}}{\vec{A}}(\vec{a})=a$, then $(\vec{A}\mapsto \vec{a})\in \marg{R^{D_0}}{\vec{A}}$. Clearly, by construction of $D_0$ it follows then that $(\vec{B}\mapsto \vec{a})\in \marg{S^{D_0}}{\vec{B}}$ and further
    that $\marg{S^{\mD}}{\vec{B}}(\vec{a})=\marg{S^{\mD_0}}{\vec{B}}(\vec{a})=a$.

    Suppose then that $\deg(\vec{A}\mapsto \vec{a})<\ar(\tau)$. We may assume that $\marg{R^{\mD}}{\vec{A}}(\vec{a})\neq 0$, as otherwise \eqref{eq:ineqind} readily follows. Then, positivity of $\Mo$ entails $(\vec{A}\mapsto \vec{a})\in \marg{R^{D}}{\vec{A}}$. Note that $\ordcl{\Sigma}$ includes all reflexivity axioms, including $\inc{R}{\vec{A}}{R}{\vec{A}}$. Therefore, there exists $t\in R^D$ such that $t(\vec{A})=\vec{a}$, and $t(C)=*$ otherwise.
    Consequently, writing $\mD^*_1\defeq \chaseplus{\mD_1}{\cl{\Sigma}}$ and using monotonicity of the $\mplus$-chase, we obtain 
    \[
    b \leq \marg{R^{\mD_1}}{\vec{A}}(\vec{a}) \leq \marg{R^{\mD^*_1}}{\vec{A}}(\vec{a}).
    \]
    Furthermore, since $\marg{R^{\mD_0}}{\vec{A}}(\vec{a})$ is of the form $ a \mplus \dots \mplus a$, it is absorpted by $b$, leading us to infer by Theorem \ref{thm:plusterminates} that
        \begin{align*} 
    &\marg{R^{\mD}}{\vec{A}}(\vec{a})= \marg{R^{\mD_0}}{\vec{A}}(\vec{a}) \mplus \marg{R^{\mD^*_1}}{\vec{A}}(\vec{a}) \leq \marg{R^{\mD^*_1}}{\vec{A}}(\vec{a})\\
    \leq &\marg{S^{\mD^*_1}}{\vec{B}}(\vec{a}) \leq \marg{S^{\mD_0}}{\vec{B}}(\vec{a}) \mplus \marg{S^{\mD^*_1}}{\vec{B}}(\vec{a}) = \marg{S^{\mD}}{\vec{B}}(\vec{a})
        \end{align*}
Hence, we have shown \eqref{eq:ineqind}. Thus $\mD \models \Sigma$, and hence Case 1 is concluded.

\vspace{2mm}
\noindent\textbf{Case 2: the natural order of $\Mo_b$ is not antisymmetric.}
Then by Lemma \ref{lem:wa}, $\Mo_b$ is eventually periodic. Let $d = mb$ be the element witnessing this periodicity as in \Cref{fig:chain}. It is clear that $d$ is greater than  any element preceding it in the chain. By transitivity of the natural order, it is also greater than any element in the eventually periodic part. We therefore obtain that $c \leq d$ for all $c \in \Mo_b$. Now define $\Mo$-databases $\mD_0$ and $\mD_1$ by setting, for each $R \in \sch$, $R^{\mD_0}(t)=b$ for all $t\in R^{D_0}$ and $R^{\mD_1}(t)=d$ for all $t\in R^{D_1}$, and mapping all other tuples to $0$.
Define
\[
\mD \defeq \mD_0 \mplus \mD_1.
\]
By positivity of $\Mo_b$, $D = \supp(\mD_0\mplus \mD_1)$. We then obtain that $\supp(\mD_0\mplus \mD_1) \not\models \tau$ by \eqref{eq:counter} and $\mD \not\models \tau$ by Proposition~\ref{prop:preserved}.

As for $\mD \models \Sigma$, consider an arbitrary inclusion dependency $\inc{R}{\vec{A}}{S}{\vec{B}} \in \Sigma$. By the same reasoning as in Case 1, \eqref{eq:ineqind} holds for the case $\deg(\vec{A}\mapsto \vec{a})=\ar(\tau)$, now with $\marg{R^{\mD_0}}{\vec{A}}(\vec{a}) \in \{0,b\}$. For the case $\deg(\vec{A}\mapsto \vec{a}) < \ar(\tau)$, since $\marg{R^{\mD_0}}{\vec{A}}(\vec{a})$ is of the form $b \mplus \dots \mplus b$ and $\marg{R^{\mD_1}}{\vec{A}}(\vec{a})$ is of the form $d \mplus \dots \mplus d$, we have that
    \begin{align*} 
        &\marg{R^{\mD}}{\vec{A}}(\vec{a}) = \marg{R^{\mD_0}}{\vec{A}}(\vec{a}) \mplus \marg{R^{\mD_1}}{\vec{A}}(\vec{a}) \leq d \\
        &\leq \marg{S^{\mD_0}}{\vec{B}}(\vec{a}) \mplus \marg{S^{\mD_1}}{\vec{B}}(\vec{a}) = \marg{S^{\mD}}{\vec{B}}(\vec{a})
    \end{align*}
Thus, $\mD \models \Sigma$ and hence, $\Sigma \not\models_{\Mo} \tau$.
\end{proof}

\subsection{Main Theorem: Weakly Absorptive Monoids}
We are now in a position to state the main theorem of this section, which unifies the two cases considered above.
The result follows directly from Theorems~\ref{prop:cwa} and~\ref{prop:fwa}.
\begin{theorem}\label{thm:wa}
    Let $\Mo$ be a weakly absorptive positive commutative monoid.
    The axioms of inclusion dependencies are sound and complete over $\Mo$-databases.
\end{theorem}
The proof of the following corollary is almost identical to that of Corollary \ref{cor:wc}.
\begin{corollary}\label{cor:wab}
    Let $\Mo$ be a weakly absorptive positive commutative monoid.
    The axioms of inclusion dependencies extended with balance axioms are sound and complete over balanced $\Mo$-databases.
\end{corollary}
\begin{proof}
Soundness follows from Prop. \ref{prop:sound}.
For completeness, suppose $\Sigma\momodels{\Mo,\rm b}\sigma$. 
Then $\Sigma^* \models_{\Mo}\sigma$, where $\Sigma^*$ extends $\Sigma$ by all balance axioms $\inc{S}{\emptyset}{R}{\emptyset}$, where $R$ and $S$ are relation names appearing in $\Sigma$ or $\sigma$.
It follows by \Cref{thm:wa} that $\sigma$ is derivable from $\Sigma^*$ using the axioms of inclusion dependencies. This entails $\sigma$ is derivable from $\Sigma$ using the axioms of inclusion dependencies extended with balance axioms.\looseness=-1
\end{proof}

\section{Conclusion}
We studied the implication problem for inclusion dependencies over $\Mo$-databases,
where tuple annotations range over a positive commutative monoid $\Mo$.
Our main result establishes a complete classification of the axiomatisability
of this problem in terms of simple algebraic properties of $\Mo$.
In particular, the standard axioms for inclusion dependencies are sound and
complete precisely when $\Mo$ is weakly absorptive, whereas for weakly
cancellative monoids one additionally needs the weak symmetry axiom.
When restricting to balanced databases, a further balance axiom is required,
yielding symmetry as a derived rule.

Technically, the proofs rely on a novel variant of the chase procedure,
the $\oplus$-chase, which generalises the classical relational chase to
monoids with a total natural order and may be of independent interest.
This chase variant terminates whenever weak symmetry is imposed, and
it yields a canonical construction for completeness in the weakly
cancellative case.

\section*{Acknowledgments}
The first author was partially supported by the ERC grant 101020762.
The third author was partially supported by the London Mathematical Society's Computer Science Small Grant.

\bibliographystyle{kr}
\bibliography{biblio}

\newpage
\appendix

\section{Proof of Lemma~\ref{lem:terminates}}
For an IND $\sigma=\inc{R}{\vec{A}}{S}{\vec{B}}$, the \emph{inverse} of $\sigma$, denoted $\inv{\sigma}$, is defined as $\inc{S}{\vec{B}}{R}{\vec{A}}$.

\setcounter{theorem}{11}
\terminates*
\setcounter{theorem}{27}

\begin{proof} 
Consider an arbitrary $\oplus$-chase sequence $(\mD_i)_{i\geq 0}$ from $\mD$ by $\Sigma$.  
Assume toward contradiction that this sequence is infinite.
Let us write $T^*=\bigcup_{i\geq 0} \{t \in R \mid R \in \supp(\mD_i)\}$ for the set of tuples  appearing in the supports of the $\Mo$-databases in the chase sequence.
Furthermore, let $T^{\infty} \subseteq T^*$ consist of those tuples that are incremented by infinitely many chase steps, and let $T^{<\infty}\coloneqq  T^{\infty}\setminus T^*$.
Then, 
fix $N \in \mathbb{N}$ such that
\[
\forall i \ge N:\quad
\chasenext{\mD_i}{\mD_{i+1}}{\sigma_i,\vec{a}_i}
\text{ increments a tuple in } T^{\infty}.
\]

Since the $\oplus$-chase does not introduce any other fresh values than $*$,  $T^*$ is finite. 
Let 
\begin{equation}\label{eq:M}
M \coloneq \max\{\deg(t)\mid t\in T^\infty\}.
\end{equation}
Next, we prove three helping claims.

\begin{claim}\label{claim:one}
    Let $i\geq N$, and let $R\in \sch$. Let $\vec{A}$ be a sequence of attributes from $\Att(R)$, and let $\vec{a}=a_1, \dots ,a_n$, $n=|\vec{A}|$, be a sequence of constants such that $|\{\ell\in [n]\mid a_\ell\neq *\}|>M$. Then $\marg{R^{\mD_i}}{\vec{A}}(\vec{a})=\marg{R^{\mD_N}}{\vec{A}}(\vec{a})$.
\end{claim}
\begin{proof}
    Chase steps of the form $\chasenext{\mD_j}{\mD_{j+1}}{\sigma_i,\vec{a}_i}$, for $j\geq N$, increment tuples from $T^{\infty}$. In particular, they do not increment $R$-tuples $t$ such that $t(\vec{A})=\vec{a}$. Thus the claim follows.
\end{proof}

\begin{claim}\label{claim:two}
    Let $\sigma=\inc{R}{\vec{A}}{S}{\vec{B}}\in \Sigma$ be such that $\inv{\sigma}\in \Sigma$. Consider a sequence $\vec{a}=a_1, \dots ,a_n$, $n=|\vec{A}|$,  such that $|\{ \ell\in [n]\mid  a_\ell\neq *\}|>M$.
    Then  
    \begin{equation}\label{eq:impliedeq}
            \marg{R^{\mD_N}}{\vec{A}}(\vec{a}) = \marg{S^{\mD_N}}{\vec{B}}(\vec{a}).
    \end{equation}
\end{claim}
\begin{proof}
    Assume toward contradiction that \eqref{eq:impliedeq} is false. %
    By antisymmetry of the natural order, and symmetry, we may assume that 
    \begin{equation}\label{eq:disequal}
            \marg{R^{\mD_N}}{\vec{A}}(\vec{a}) \not\leq \marg{S^{\mD_N}}{\vec{B}}(\vec{a}).
    \end{equation}
    By \cref{it:three} of the $\oplus$-chase sequence, we find $j> N$ such that 
    \[
    \marg{R^{\mD_N}}{\vec{A}}(\vec{a}) \leq %
    \marg{S^{\mD_j}}{\vec{B}}(\vec{a}).\]
    This is only possible when some chase step
    $\chasenext{\mD_i}{\mD_{i+1}}{\sigma_i,\vec{a}_i}$ with $i\geq N$ increments a tuple in $T^{<\infty}$. This creates  a contradiction with the definition of $N$.
    Hence \eqref{eq:impliedeq} is true.
\end{proof}

Hereafter, when writing INDs in a two-partite form $\inc{R}{\vec{A}\vec{C}}{S}{\vec{B}\vec{D}}$,
we tacitly assume that $|\vec{A}| = |\vec{B}|=M$ for $M$ fixed in \eqref{eq:M}, and
$|\vec{C}| = |\vec{D}|$.

\begin{claim}\label{claim:three}
Let an IND $\sigma = \inc{R}{\vec{A}\vec{C}}{S}{\vec{B}\vec{D}}\in \Sigma$ be such that $\inv{\sigma}\in \Sigma$.
Consider a sequence $\vec{a}\in (\Const\setminus \{*\})^M$. Let $i\geq N$. 
Suppose the chase sequence contains a step $\chasenext{\mD_i}{\mD_{i+1}}{\sigma,\vec{a}\vec{*}}$. 
Then,
    \begin{equation}\label{eq:equal}
            \marg{R^{\mD_i}}{\vec{A}}(\vec{a}) = \marg{S^{\mD_{i+1}}}{\vec{B}}(\vec{a}).
    \end{equation}
\end{claim}
\begin{proof}
We obtain \eqref{eq:equal} via the following chain of equalities and inequalities:
\begin{align}
    &\marg{R^{\mD_i}}{\vec{A}}(\vec{a}) \nonumber\\%
    &= \marg{R^{\mD_i}}{\vec{A}\vec{C}}(\vec{a}\vec{*})  \mplus \bigmplus_{\vec{b}\neq \vec{*}} \marg{R^{\mD_i}}{\vec{A}\vec{C}}(\vec{a}\vec{b})  \nonumber\\
&= \marg{R^{\mD_i}}{\vec{A}\vec{C}}(\vec{a}\vec{*}) \mplus \bigmplus_{\vec{b}\neq \vec{*}} \marg{R^{\mD_N}}{\vec{A}\vec{C}}(\vec{a}\vec{b})\label{eq:m}\\
&= \marg{S^{\mD_{i+1}}}{\vec{A}\vec{C}}(\vec{a}\vec{*}) \mplus \bigmplus_{\vec{b}\neq \vec{*}} \marg{S^{\mD_N}}{\vec{B}\vec{D}}(\vec{a}\vec{b})\label{eq:n}\\
&= \dots = \marg{S^{\mD_{i+1}}}{\vec{B}}(\vec{a}).\label{eq:o}
\end{align}
Note that \eqref{eq:m} is by Claim \ref{claim:one}. Furthermore, \eqref{eq:n} follows by 
 \eqref{eq:equals} and Claim \ref{claim:two}. %
The remaining equalities leading to \eqref{eq:o} are derived symmetrically.
\end{proof}

Let us now create a sequence $(i_j)_{j\geq 0}$, $i_0=N$, such that it satisfies the following property:
\begin{quote}
(\emph{Closure property})
If the $\oplus$-rule is applicable to $\mD_{i_j}$ with respect to
$\sigma = \inc{R}{\vec{A}}{S}{\vec{B}} \in \Sigma$ and $\vec{a}$, then
\[
\marg{R^{\mD_{i_j}}}{\vec{A}}(\vec{a})
\;\leq\;
\marg{S^{\mD_{i_{j+1}}}}{\vec{B}}(\vec{a}).
\]
\end{quote}
 Informally, any IND that is violated at level $i_j$ is ``closed'' at level $i_{j+1}$. The existence of $(i_j)_{j\geq 0}$ with the closure property is guaranteed by \cref{it:three} of the $\oplus$-chase sequence.

Hereafter, by slight abuse of notation
$\chasenext{\mD}{\mD'}{\sigma,\vec{a}}$
can denote any chase step $\chasenext{\mD}{\mD'}{\sigma',\vec{a}'}$ where $\sigma'$ and $\vec{a}'$ are obtain from 
$\sigma$ and $\vec{a}$, respectively, by the same permutation.
Then, since some tuple $t$ with degree $M$ (defined in \eqref{eq:degree}) is incremented infinitely many times during the chase, for any $j\in \mathbb{N}$ we find a chase step of the form $\chasenext{\mD_{\ell-1}}{\mD_{\ell}}{\sigma,\vec{a}\vec{*}}$, where $\sigma=\inc{R}{\vec{A}\vec{C}}{S}{\vec{B}\vec{D}}$, $\vec{a} \in (\Const \setminus \{*\})^M$, and $\ell>  i_j$. Provided that $j>0$, we can show that the multiplicity $\marg{R^{\mD_{\ell-1}}}{\vec{A}}(\vec{a})$ was  incremented previously at some level $\ell'>  i_{j-1}$.
\begin{claim}\label{claim:four}
   Let $j\in \mathbb{N}\setminus \{0\}$ and $\ell> i_j$. Suppose $\chasenext{\mD_{\ell-1}}{\mD_{\ell}}{\sigma,\vec{a}\vec{*}}$, where $\sigma=\inc{R}{\vec{A}\vec{C}}{S}{\vec{B}\vec{D}}$ and $\vec{a}\in (\Const \setminus \{*\})^M$. %
   Then, $\chasenext{\mD_{\ell'-1}}{\mD_{\ell'}}{\sigma',\vec{a}\vec{*}}$, for some
   $\ell'<\ell$ and $\sigma'=\inc{U}{\vec{E}\vec{F}}{R}{\vec{A}\vec{C}}$
   such that
   \begin{itemize}
        \item $ \ell' > i_{j-1}$ and
       \item $\marg{R^{\mD_{\ell'}}}{\vec{A}}(\vec{a})=\marg{R^{\mD_{\ell-1}}}{\vec{A}}(\vec{a})$, and
   \end{itemize}
\end{claim}
\begin{proof} 
We define $\ell'< \ell$ as the least integer with $\marg{R^{\mD_{\ell'}}}{\vec{A}}(\vec{a})=\marg{R^{\mD_{\ell-1}}}{\vec{A}}(\vec{a})$. Consequently, $\marg{R^{\mD_{\ell'}}}{\vec{A}\vec{C}}(\vec{a}\vec{*}) = \marg{R^{\mD_{\ell-1}}}{\vec{A}\vec{C}}(\vec{a}\vec{*})$. We claim that $\ell' > i_{j-1}$. Assume toward contradiction that $\ell' \leq i_{j-1}$. 
By the closure property of $(i_j)_{j\geq 0}$, 
we obtain
\begin{align*} 
\marg{R^{\mD_{\ell-1}}}{\vec{A}\vec{C}}(\vec{a}\vec{*})=& \marg{R^{\mD_{\ell'}}}{\vec{A}\vec{C}}(\vec{a}\vec{*})
\leq \marg{R^{\mD_{i_{j-1}}}}{\vec{A}\vec{C}}(\vec{a}\vec{*})\\
\leq &\marg{S^{\mD_{i_{j}}}}{\vec{B}\vec{D}}(\vec{a}\vec{*})   
 \leq \marg{S^{\mD_{\ell-1}}}{\vec{B}\vec{D}}(\vec{a}\vec{*}).
\end{align*}
This entails $\mplus$-rule is not applicable to $\mD_{\ell-1}$ w.r.t. $\sigma,\alpha$, contradicting the assumption that $\chasenext{\mD_{\ell-1}}{\mD_{\ell}}{\sigma,\vec{a}\vec{*}}$. Hence the claim that $\ell' > i_{j-1}$ holds.

It follows that $\ell'> N$. Moreover, by construction we have $\marg{R^{\mD_{{\ell'}-1}}}{\vec{A}}(\vec{a})\neq \marg{R^{\mD_{\ell'}}}{\vec{A}}(\vec{a})$. Hence there exists a chase step of the form $\chasenext{\mD_{{\ell'}-1}}{\mD_{\ell'}}{\tau,\vec{a}\vec{*}}$, for some $\sigma'=\inc{U}{\vec{E}\vec{F}}{R}{\vec{A}\vec{C}}$.  This proves the claim.
\end{proof}

We can now finish the proof of the lemma.
First, by repeated application of Claim \ref{claim:four} (starting from $\ell_n>i_j$ for large enough $j$), we find integers $N<\ell_1< \dots <\ell_n$, relation names $R_1, \dots ,R_n \in \Rel$, and INDs $\sigma_{0}, \dots ,\sigma_{n}\in \Sigma$ (where, for each $i\in [2,n]$, $\sigma_{i}$ is of the form $\inc{R_{i-1}}{\vec{A}_{i-1}\vec{C}_{i-1}}{R_{i}}{\vec{A}_{i}\vec{C}_{i}}$) such that $\sigma_1=\sigma_n$ and
   \begin{enumerate}
   \item\label{it:eka} $\chasenext{\mD_{\ell_{1}-1}}{\mD_{\ell_{1}}}{\sigma_{{1}},\vec{a}\vec{*}}$,
   \item $\marg{R_1^{\mD_{\ell_{1}}}}{\vec{A}_1}(\vec{a})=\marg{R_1^{\mD_{\ell_{2}-1}}}{\vec{A}_1}(\vec{a})$,
   \item\label{it:kolmas} $\chasenext{\mD_{\ell_{2}-1}}{\mD_{\ell_{2}}}{\sigma_{{2}},\vec{a}\vec{*}}$,
   \item $\dots$,
       \item\label{it:viides} $\marg{R_{n-1}^{\mD_{\ell_{n-1}}}}{\vec{A}_n}(\vec{a})=\marg{R_{n-1}^{\mD_{\ell_{n}-1}}}{\vec{A}_n}(\vec{a})$, and
       \item\label{it:vika} $\chasenext{\mD_{\ell_{n}-1}}{\mD_{\ell_{n}}}{\sigma_{n},\vec{a}\vec{*}}$.
   \end{enumerate}
In particular, the $\oplus$-rule is applied w.r.t. the same pair $\sigma,\vec{a}\vec{*}$ on \Cref{it:eka} and \Cref{it:vika}.  
Note that the INDs $\sigma_2,\dots ,\sigma_n$ form a cycle of references.
Recall also $\Sigma$ is closed under $\wsvdash$.
Hence by projection and transitivity we obtain $\inc{R_{i}}{\emptyset}{R_{i-1}}{\emptyset}\in \Sigma$, for $i\in [2,n]$. Consequently,  by weak symmetry $\inv{(\sigma_i)}\in \Sigma$, for $i\in [2,n]$. It then follows by Claim \ref{claim:three} that
\[
\marg{R_{i-1}^{\mD_{\ell_{i}-1}}}{\vec{A}_{i-1}}(\vec{a}) = \marg{R_{i}^{\mD_{\ell_{i}}}}{\vec{A}_{i}}(\vec{a}),
\]
for $i\in [2,n]$.
For example, in the previous list \cref{it:kolmas,it:vika} give respectively rise to inequalities
\begin{align*}
&\marg{R_{1}^{\mD_{\ell_{2}-1}}}{\vec{A}_{1}}(\vec{a}) = \marg{R_{2}^{\mD_{\ell_{2}}}}{\vec{A}_{2}}(\vec{a}),
\text{ and}\\
&\marg{R_{n-1}^{\mD_{\ell_{n}-1}}}{\vec{A}_{n-1}}(\vec{a}) = \marg{R_n^{\mD_{\ell_{n}}}}{\vec{A}_n}(\vec{a})
\end{align*}
Note that we can write $R=R_1=R_n$ and $\vec{A}=\vec{A}_1=\vec{A}_n$ since $\sigma_1=\sigma_n$. Hence we conclude that
\[
\marg{R^{\mD_{\ell_{1}}}}{\vec{A}}(\vec{a}) = \marg{R^{\mD_{\ell_{n}}}}{\vec{A}}(\vec{a})
\]
In particular, $\ell_1 < \ell_n$, and the $\ell_n$th application of the chase increments an $R$-tuple $t$ with $t(\vec{A})=(\vec{a})$. Hence, if we denote $a=\marg{R^{\mD_{\ell_{1}}}}{\vec{A}}(\vec{a})$, and let $b_1, \dots ,b_k\neq 0$ correspond to the incrementations of tuples $t$ with $t(\vec{A})=(\vec{a})$ in the chase steps leading from $\ell_1$ to $\ell_n$, we obtain $k\geq 1$ and  $a \mplus b_1\mplus \dots \mplus b_k = a$.  
Then weak cancellativity entails $b_1 \mplus \dots \mplus b_k=0$, and positivity entails $b_1=\dots =b_k=0$. This creates a contradiction, by which we obtain that the $\oplus$-chase sequence is finite. This concludes the proof.
\end{proof}

\section{Proof of Lemma \ref{lem:wcembed}}

We first prove the following proposition.
\begin{proposition}\label{prop:antisymmetric}
The natural order of a positive weakly cancellative commutative monoid is antisymmetric, i.e., it is a partial order.
\end{proposition}
\begin{proof}
    Let $a,b$ be elements of the monoid such that $a\leq b$ and $b\leq a$.
    We need to show that $a=b$. By the defininition of a natural order there exist some elements $c,d$ such that 
    \[
    a\mplus c=b \text{ and }b\mplus d = a.
    \]
    By commutativity,
    \(
    a = (a \mplus c) \mplus d = a \mplus (c \mplus d).
    \)
    Now, weak cancellativity entails $c\mplus d=0$, and positivity entails $c=d=0$, leading to $a=b$.
\end{proof}

\setcounter{theorem}{14}
\WCEmbed*
\setcounter{theorem}{32}
\begin{proof}
    Let us write $\Mo=(K,\mplus,0)$.   
    Since $\Mo$ is non-trivial, there exists $b \in K \setminus \{0\}$. We define $f\colon \mathbb{N} \to \Mo$ recursively by setting $f(0) \dfn 0_{\Mo}$ and $f(n + 1) \dfn f(n) \mplus b$ for all $n \geq 0$. %

    We show first that $f$ is a monoid homomorphism. Note that, by construction, $f(0)=0_\Mo$. We now show that $f(m + n) = f(m) \mplus f(n)$ for all $m, n \in \N$, which we prove by induction on $n$ (with $m$ arbitrary).

	\noindent\textbf{Base step.}\,
    For $n = 0$, the following chain of equalities is immediate:
	\[
    f(m + 0) = f(m) = f(m) \mplus 0_\Mo = f(m) \mplus f(0).
    \]

    \noindent\textbf{Induction step.}\,
	Assume the claim holds for $n=k$. Then,
	\begin{align*}
	    &\,f(m + (k + 1)) = f((m + k) + 1) = f(m + k) \mplus b\\
        =&\,(f(m) \mplus f(k)) \mplus b=(f(m) \mplus f(k)) \mplus b \\
        =&\, f(m) \mplus (f(k) \mplus b) = f(m) \mplus f(k + 1)
        \end{align*}
        by the induction hypothesis, associativity, and the construction of $f$. This concludes the induction proof.

    Next, we show that $f$ is injective. 
	Assume that $f(m) = f(n)$ for some $m, n \in \N$. By symmetry, we may assume that $m \geq n$. Let $k = m - n \geq 0$. Since $f$ is a homomorphism, $f(m) = f(n + k) = f(n) \mplus f(k)$. Hence $f(n) = f(n) \mplus f(k)$, and furthermore $f(k) = 0_\Mo$ by weak cancellativity. If $k\neq 0$, then $0_\Mo = f(k)=f((k-1) + 1) =  f(k-1)\mplus b$ by the construction of $f$. By positivity $b=0$, which is a contradiction. Hence $k=0$ and $m=n$, proving that $f$ is an injection.

      Finally, we establish that $a\leq b$ if and only if $f(a)\leq f(b)$. Suppose first $a\leq b$. Then there exists $c\in \mathbb{N}$ such that $a+c=b$. Since $f$ is a monoid homomorphism, we obtain $f(a)\mplus f(c)= f(a+c)=f(b)$ implying $f(a)\leq f(b)$. Suppose then $a \not\leq b$. 
      Since the natural order of the natural numbers is total, we have $b \leq a$, and hence there exists $c$ such that $b+c=a$. 
      Since $f$ is a monoid homomorphism, we get $f(b)\mplus f(c)=f(b+a)=f(a)$ implying $f(b) \leq f(a)$. Now, assume toward a contradiction that $f(a)\leq f(b)$. Since the natural order of $\Mo$ is antisymmetric by Proposition \ref{prop:antisymmetric}, we obtain $f(a)=f(b)$. By injectivity of $f$, this entails $a=b$, contradicting the assumption. Thus $f(a)\not \leq f(b)$. 
\end{proof}

\section{Proofs for Lemmas \ref{lem:wa} and \ref{lem:supppreserved}}

\setcounter{theorem}{22}
\WeakAbsorption*
\setcounter{theorem}{32}
\begin{proof}
    (\Cref{it:one,it:one-half})
    Since $\Mo$ is WA, choose some non-zero $a_0$ and $b_0$ such that $a_0 \mplus b_0 = b_0$. Check $b_0$: If for all $c$, $b_0 \mplus c \neq c$, we are done; take $a = a_0$, $b = b_0$. If there exists $c$ with $b_0 \mplus c = c$, then $(b_0, c)$ is another WA pair. Rename $a_1 = b_0$, $b_1 = c$.
		
		Now check $b_1$: If for all $c$, $b_1 \mplus c \neq c$, done: $a = a_1, b = b_1$. Else, there is $d$ with $b_1 \mplus d = d$, so $(b_1, d)$ is a WA pair. Rename $a_2 = b_1$, $b_2 = d$. Continue.  
		
		If this process never stops, we produce $a_0 \neq 0, a_1, a_2, \dots$ with $a_i \mplus a_{i + 1} = a_{a + 1}$ for all $i \geq 0$, which is the CA property. Since $\Mo$ is not CA, this cannot happen. Thus, the process must stop at some $k$ and we find $a_k, b_k \neq 0$ with $a_k \mplus b_k = b_k$ and no $c$ satisfies $b_k \mplus c = c$. Then, $a = a_k$ and $b = b_k$ are the desired elements.

     (\Cref{it:ntwcpc}) The properties listed here are inherited from $\Mo$.   
     
    (\Cref{it:two})
        Since the natural order readily satisfies the totality axiom, it is total once it satisfies antisymmetry.
        Assuming $\Mo_b$ is equipped with antisymmetric natural order, we prove it is weakly cancellative. 
        Suppose $nb \mplus mb = mb$ for some $n,m\in \N$.  We need to show that $n=0$.
        Assume toward contradiction that $n\geq 1$.
        Note that by the definition of the natural order, $kb\leq (k+1)b$ for all $k\in \N$.
        We obtain
        \[mb \leq (m+1)b\leq \dots \leq (m+n)b = nb \mplus mb \leq mb,\]
        where the last inequality is due to reflexivity of the natural order. By transitivity and antisymmetry of the natural order, this entails
        $ b\mplus mb=(m+1)b =mb$, contradicting item \ref{it:one}. By contradiction, we obtain $n=0$, concluding the proof.

    (\Cref{it:two-half})     %
    Assuming the natural order of $\Mo_b$ is not antisymmetric, we obtain that there exist distinct $kb$ and $nb$ such that $kb \leq nb$ and $nb \leq kb$ for some $k, n \in \N$. Observe that by positivity of $\Mo_b$, both $k$ and $n$ are non-zero. By definition of the natural order $(k + r)b = nb$ and $(n + s)b = kb$ for some $r, s \geq 1$. Then we have that $(k + r + s)b = kb$, i.e. $kb = (k + t)b$ where $t = r + s \geq 2$.
    This concludes the proof.
\end{proof}

\setcounter{theorem}{23}
\SuppPreserved*
\setcounter{theorem}{32}
\begin{proof}
Assume $\chase{\supp(\mD)}{\Sigma})=\supp(\mD)$. 
By Theorem \ref{thm:plusterminates}, the $\chaseplus{\mD}{\Sigma}$ is a $\Mo$-database satisfying $\Sigma$. Let $(\mD_0,\dots ,\mD_n$ the $\mplus$-chase sequence of $\mD$ by $\Sigma$
such that $\mD_n=\chaseplus{\mD}{\Sigma}$. It suffices to prove by induction on $i$ that $\supp(\mD_i)=\supp(\mD)$.

The base step is immediate. For the induction step, consider the induction hypothesis that $\supp(\mD_i)=\supp(\mD)$. Now, $\chasenext{\mD_i}{\mD_{i+1}}{\sigma,\vec{a}}$ for some $\sigma\in \Sigma$ and $\vec{a}\in \Const^{\ar(\sigma)}$.
Now, if $\sigma$ is of the form $\inc{R}{\vec{A}}{S}{\vec{B}}\in \Sigma$ and $\vec{a}\in \Const^{\ar(\sigma)}$, we have $\marg{R^{\mD_i}}{\vec{A}}(\vec{a}) \neq 0$. Then $R^{\supp(\mD_i)}$ contains a tuple $t$ such that $t(\vec{A})=\vec{a}$.
By the induction hypothesis and the assumption, $S^{\supp(\mD_i)}$ contains a tuple $t'$ such that $t'(\vec{B})=\vec{a}$, and $t'(C) = *$ otherwise. By definition, $\supp(\mD_{i+1})$ is otherwise as $\supp(\mD_{i})$, except that $S^{\supp(\mD_{i+1})}=S^{\supp(\mD_{i})}\cup\{t'\}$. Hence, we conclude that $\supp(\mD_{i+1})=\supp(\mD_{i})$, which proves the induction step.
\end{proof}

\end{document}